\documentclass[reqno,centertags,11pt]{amsart}
\usepackage{amsmath,amsthm,amscd,amssymb} \usepackage{latexsym}
\usepackage{graphicx}
\usepackage{color}

 \newcommand{\N}{{\mathbb{N}}}
 \newcommand{\R}{{\mathbb{R}}}

 \newcommand{\Z}{{\mathbb{Z}}}

 \newcommand{\f}{\frac}
 \newcommand{\ol}{\overline}
 \newcommand{\wti}{\widetilde  }

\newcommand{\hatt}{\widehat} 
\newcommand{\beq}{\begin{equation}}
\newcommand{\eeq}{\end{equation}}
\newcommand{\bdm}{\begin{displaymath}}
\newcommand{\edm}{\end{displaymath}} \newcommand{\ba}{\begin{align}}
\newcommand{\ea}{\end{align}} \newcommand{\bpf}{\begin{proof}}
\newcommand{\epf}{\end{proof}}

\newcommand{\la}{\langle} \newcommand{\ra}{\rangle}

 \allowdisplaybreaks
\newcommand{\mathcalB}{\mathcal{B}}

\newcommand{\mathcalD}{\mathcal{D}}

\newcommand{\mathcalG}{\mathcal{G}}

\newcommand{\mathcalS}{{\mathcal{S}}}

\newtheorem{thm}{Theorem} \newtheorem{prop}[thm]{Proposition}
\newtheorem{lem}[thm]{Lemma} 
\newtheorem{cor}[thm]{Corollary}

\theoremstyle{definition}

 \newtheorem{remark}[thm]{Remark}

\newcounter{theoremi}[thm] 

\numberwithin{thm}{section} \numberwithin{equation}{section}

\begin{document}

\title[Ballistic Transport]{Ballistic Transport for the
Schr\"odinger Operator with Limit-Periodic or Quasi-periodic Potential in Dimension Two}
\author[Y.~Karpeshina, Y.-R.~Lee, R.~Shterenberg, G.~Stolz ]{Yulia Karpeshina, Young-Ran~Lee, Roman Shterenberg and G\"unter Stolz}

%Dr. Karpeshina's address
\address{Department of Mathematics, Campbell Hall, University of Alabama at Birmingham,
1300 University Boulevard, Birmingham, AL 35294.}
\email{karpeshi@uab.edu}%

%Young-Ran's address
\address{Department of Mathematics, Sogang University, 5 Baekbeom-ro,
    Mapo-gu, Seoul, 121-742, South Korea.}%
\email{younglee@sogang.ac.kr}

%Roman's address
\address{Department of Mathematics, Campbell Hall, University of Alabama at Birmingham,
1300 University Boulevard, Birmingham, AL 35294.}
\email{shterenb@math.uab.edu}

%Gunter's address
\address{Department of Mathematics, Campbell Hall, University of Alabama at Birmingham,
1300 University Boulevard, Birmingham, AL 35294.}
\email{stolz@uab.edu}

\thanks{Supported in part by NSF-grants DMS-1201048 (Y.K.) and DMS-1069320 (G.S.),  National Research Foundation of Korea (NRF) grants funded by the Korea government (MSIP) No.\ 2011-0013073 and (MOE) No.\ 2014R1A1A2058848 (Y.-R.L.) and Simons Foundation grant No. 312879 (R.S)}
\thanks{The authors are thankful to the Isaac Newton Institute for Mathematical Sciences (Cambridge University, UK) for support and hospitality during the program Periodic and Ergodic Spectral Problems where work on this paper was undertaken.}

%\thanks{\copyright 2007 by the authors. Faithful reproduction of this article,
%       in its entirety, by any means is permitted for non-commercial purposes}
%\keywords{Non-linear Schr\"odinger equation, decay of eigenfunctions}
%\subjclass[2000]{35B20, 35B40, 35P30 }
\date{\today}

\begin{abstract}
We prove the existence of ballistic transport for the
Schr\"odinger operator with limit-periodic or quasi-periodic potential in dimension two. This is done under certain regularity assumptions on the potential which have been used in prior work to establish the existence of an absolutely continuous component and other spectral properties. The latter include detailed information on the structure of generalized eigenvalues and eigenfunctions. These allow to establish the crucial ballistic lower bound through integration by parts on an appropriate extension of a Cantor set in momentum space, as well as  through stationary phase arguments.
\end{abstract}

\maketitle

%%%%%%%%%%%%%%%%%%%%%%%%%%%%%%%%%%%%%%%%%%%%%%%%%%%%%%%%%%%%%%%%%%%%%%%%%%%%%%%%%%%%%%%%%%%%%%%%

\section{Introduction}

\subsection{Prior results on ballistic transport} \label{prior}

It is well known that the spectral and dynamical properties of Schr\"odinger operators $H=-\Delta+V$, either discrete in ${\mathcal H} = \ell^2(\Z^d)$ or continuous in ${\mathcal H} = L^2(\R^d)$, are related. A general correspondence of this kind is given by the RAGE theorem, e.g.\ \cite{ReedSimon}: Stated briefly, it says that solutions $\Psi(\cdot,t) = e^{-iHt} \Psi_0$ of the time-dependent Schr\"odinger equation are `bound states' if the spectral measure $\mu_{\Psi_0}$ of the initial state $\Psi_0$ is pure point, while $\Psi(\cdot,t)$ is a `scattering state' if $\mu_{\Psi_0}$ is (absolutely) continuous. However, knowing the spectral type is not sufficient to quantify transport properties more precisely, for example in terms of diffusion exponents $\beta$. The latter, if they exist, characterize how time-averaged $m$-moments ($m>0$)
\begin{equation} \label{moments}
\langle \langle X_{\Psi_0}^m \rangle \rangle_T := \frac{2}{T} \int_0^{\infty} \exp\left(-\frac{2t}{T}\right) \| X^{m/2} \Psi(\cdot,t)\|^2_{\mathcal H} \,dt
\end{equation}
of the position operator $X$ grow as a power $T^{m\beta}$ of time $T$, where $(Xu)(x) = |x|u(x)$. The special cases $\beta = 1$, $\beta=1/2$ and $\beta=0$ are interpreted as ballistic transport, diffusive transport, and dynamical localization, respectively.

In general, due to the possibility of fast traveling small tails, $\beta$ may depend on $m$. Here we will restrict our attention to the most frequently considered case of the second moment $m=2$. The ballistic upper bound
\begin{equation} \label{genballistic}
\| X \Psi(\cdot,t) \|^2_{\mathcal H}  \le C_1(\Psi_0) T^2 + C_2(\Psi_0),
\end{equation}
and thus also its averaged version $\langle \langle X_{\Psi_0}^2 \rangle \rangle_T \le C_1(\Psi_0) T^2 + C_2(\Psi_0)$, is known to hold for general potentials $V$ with relative $\Delta$-bound less than one (in particular all bounded potentials) and initial states
\begin{equation} \label{RScond}
\Psi_0 \in \mathcal{S}_1:= \{f\in L^2: |x|f \in L^2,\, |p|f \in L^2\},
\end{equation}
see \cite{RaSi}. Here $p$ is the momentum operator, i.e.\ multiplication by the variable in momentum space. As most authors, we will work here with the Abel mean used in \eqref{moments}, but mention that the existence of a ballistic upper bound can be used to show that Abel means and Cesaro means $T^{-1} \int_0^T \ldots\,dt$ lead to the same diffusion exponents (see for example Theorem~2.20 in \cite{DT2010}).

In the late 1980s and 1990s methods were developed which led to more concrete bounds on diffusion exponents by also taking fractal dimensions of the associated spectral measures into account and showing that this gives lower transport bounds. In particular, again for the special case of the second moment, the Guarneri-Combes theorem \cite{G,G2,C,L} says that
\begin{equation} \label{GuarneriCombes}
\langle \langle X_{\Psi_0}^2 \rangle \rangle_T \ge C_{\Psi_0} T^{2\alpha/d}.
\end{equation}
for initial states $\Psi_0$ with uniformly $\alpha$-H\"older continuous spectral measure (and satisfying an additional energy bound in the continuum case \cite{C}). In dimension $d=1$ this says that states with an absolutely continuous spectral measure ($\alpha=1$) also will have ballistic transport (as by \eqref{genballistic} the transport can not be faster than ballistic). In particular, this means that in cases where the spectra of one-dimensional Schr\"odinger operators with limit or quasi-periodic potentials were found to have an a.c.\ component, e.g.\ \cite{AS,CD,DS,El,MC,MP,PF}, one also gets ballistic transport.

The bound (\ref{GuarneriCombes}) does not suffice to conclude ballistic transport from the existence of a.c.\ spectrum in dimension $d\ge 2$. In fact, examples of Schr\"odinger operators with absolutely continuous spectrum, but slower than ballistic transport have been found: A two-dimensional  `jelly-roll' example with a.c.\ spectrum and diffusive transport is discussed in \cite{KL}, while \cite{BSB} provides examples of separable potentials in dimension $d\ge 3$ with a.c.\ spectrum and sub-diffusive transport.

In general, growth properties of generalized eigenfunctions have to be used in addition to spectral information for a more complete characterization of the dynamics. General relations between eigenfunction growth and spectral type as well as dynamics were found in  \cite{KL}. A series of works studied one-dimensional models with $\alpha<1$ and related the dynamics to transfer matrix bounds, e.g.\ \cite{DLS,DT,DT2,DT3,GKT,JSBS,Tch}. In particular, these methods can establish lower transport bounds in models with sub-ballistic transport, such as the Fibonacci Hamiltonian and the random dimer model.

Much less has been done for $d\ge 2$.  Ballistic lower bounds and thus the existence of waves propagating at non-zero velocity are known only for $V=0$, where this is classical, e.g.\ \cite{ReedSimon}, and for periodic potentials \cite{AK}. Scattering theoretic methods show that this extends to potentials of sufficiently rapid decay, or sufficiently rapidly decaying perturbations of periodic potentials. However, to our knowledge there are no prior results on ballistic lower bounds for multidimensional Schr\"odinger operators with bounded potentials which are not asymptotically periodic. Providing two such results in dimension $d=2$, one for a class of limit-periodic potentials and one for a class of quasi-periodic potentials, is our main goal here.

For both of these examples, the existence of an absolutely continuous component in the spectrum has been shown in earlier works \cite{KL3,KS,KS1}. Essentially, what we do here is to show that the properties of generalized eigenfunctions which were obtained in these works can be used to also conclude ballistic transport.

%%%%%%%%%%%%%%%%%%%%%%%%%%%%%%%%%%%%%%%%%%%%%%%%%%%%%%%%%%%%%%%%%%%%%%%%%%%%%%%%%%%%%%%%%%%%%%%%

\subsection{Models and Assumptions}

We study the initial value problem
 \begin{equation}\label{IVP}
   i \frac{\partial \Psi}{\partial t}=H\Psi, \ \ \ \Psi (\vec x,0)=\Psi _0(\vec x)
    \end{equation}
for  the Schr\"odinger operator
    \begin{equation}\label{limper}
    H=-\Delta+V(\vec x)
    \end{equation}
    in two dimensions, $\vec x \in \R^2$. For the potential $V(\vec x)$ we consider two cases, a class of limit-periodic potentials and a class of quasi-periodic potentials.

\subsubsection{Limit periodic case}   Here we assume that the potential can be written as
    \begin{equation}\label{V}
    V(\vec x)=\sum _{r=1}^{\infty}V_r(\vec x),
    \end{equation}
   where $\{V_r\}_{r=1}^{\infty}$ is a family of periodic potentials
with doubling periods. More precisely, $V_r$ has orthogonal periods $2^{r-1}\vec{d_1},\ 2^{r-1}\vec{d _2}$.
Without loss of generality, we assume that  $\vec
d_1=(d_1,0)$, $\vec d_2=(0,d_2)$ and $\int _{Q_r}V_r(\vec x)d\vec x=0$, where $Q_r = [0, 2^{r-1}d_1] \times [0, 2^{r-1} d_2]$
is the elementary cell of periods corresponding to $V_r$. We also
assume that all $V_r$ are real trigonometric polynomials with the
lengths growing at most linearly in the  period. Namely, there exists
a positive number $R_0 <\infty$ such that each potential admits the
Fourier representation
    \beq\label{1.2a}
    V_r(\vec x)=\sum _{q \in  \Z^2\setminus \{0\},\ 2^{-r+1}|q|<R_0} v_{r,q}
    e^{ i \la 2^{-r+1}\tilde q,\vec x \ra} ,\ \ \ \tilde q=2\pi
    \left(\frac{q_1}{d_1}, \frac{q_2}{d_2}\right),
    \eeq
$\langle \cdot ,\cdot\rangle $ being the canonical scalar product and $|\cdot|$ the corresponding norm
in $\R^2$.  We assume that the series \eqref{V} converges super-exponentially fast:
\beq \label{be}
    \sum_{q} |v_{r,q}|<\hat C\exp(-2^{\eta r})
    \eeq
for some $\eta>\eta_0 >0$ uniform in $r$. Without loss of generality we can set $\hat C=1$.

\subsubsection{Quasi-periodic case} We assume that $V$ is real and can be written in the form
\begin{equation}\label{Vq}
V(\vec x)=\sum\limits_{   {\bf s}_1,   {\bf s}_2\in\Z^2,\,   {\bf s}_1+\alpha   {\bf s}_2\in{ \mathcal
S}} V_{   {\bf s}_1,   {\bf s}_2}e^{2\pi i\langle    {\bf s}_1+\alpha    {\bf s}_2,\vec x\rangle},
\end{equation} where  $\alpha$ is an irrational number in $(0,1)$ and ${ \mathcal S}$ a finite subset of $\R^2$. Note that $V$ real means that ${\mathcal S}$ is symmetric with respect to ${\bf 0}$ and $V_{-{\bf s}_1, -{\bf s}_2} = \overline{V_{   {\bf s}_1,   {\bf s}_2}}$.

Two additional technical conditions, one on $\alpha$ and one on ${\mathcal S}$, are used to avoid certain degeneracies:

\begin{description}

\item[{\bf A1}] There are $N_0,N_1>0$
 such that if $|n_1|+|n_2|+|n_3|>N_1$ then
\begin{equation}
n_1+\alpha n_2+\alpha^2 n_3=0\ \ \  \hbox{or}\ \ \ |n_1+\alpha
n_2+\alpha^2 n_3|> (|n_1|+|n_2|+|n_3|)^{-N_0}. \label{condition}
\end{equation}

\item[{\bf A2}] If ${\bf s}\in {\mathcal S}$ and ${\bf r} \in {\mathcal S}$ have the same direction, i.e.
\begin{equation}
{\bf s} =c_* {\bf r} \quad \mbox{for some $c_* \in \R$},
\end{equation}
then $c_*$ is rational.

\end{description}

Let us discuss the meaning of these assumptions and argue that they are not very restrictive:

The condition {\bf A1} says that $\alpha$ has diophantine properties in the sense of not being  `too well' approximated by rationals. In fact, the special case $n_3=0$ of {\bf A1} implies that $\alpha$ is {\it not} a Liouville number, meaning that it has irrationality measure $\mu<\infty$  (defined as the smallest $\mu$ such that $|\alpha-p/q| > q^{-\mu-\varepsilon}$ for all $\varepsilon>0$ and integers $p$ and $q$ with $q$ sufficiently large). Note also that $\mu\geq 2$ for any irrational number $\alpha$.

We may thus interpret {\bf A1} as a {\it strong (or quadratic) non-Liouville property} of $\alpha$. In \cite{KS}, whose results we use here, it is used as a technical condition which allows to estimate the angle between two non-colinear vectors ${\bf s}_1+\alpha {\bf s}_2$ and ${\bf s}_1'+\alpha {\bf s}_2'$ (with ${\bf s}_1, {\bf s}_ 2, {\bf s}_1', {\bf s}_2' \in \Z^2$) from below by a negative power of $|{\bf s}_1+\alpha {\bf s}_2| +|{\bf s}_1'+\alpha {\bf s}_2'|$.

The condition {\bf A1} holds, in particular, for quadratic irrationals. To see this, consider a non-trivial triple $(n'_1, n'_2,n'_3)$ of integers such that
\begin{equation}
n'_1+\alpha
n'_2+\alpha^2 n'_3=0. \label{quadirr}
\end{equation}
This triple is unique up to trivial multiplication (otherwise $\alpha$ is rational). If $n_1+\alpha
n_2+\alpha^2 n_3\not=0$ for some other triple $(n_1, n_2,n_3)$, then $\mu=2$ (true for all algebraic irrationals) can be used to show that \eqref{condition} holds with $N_0=2+\varepsilon$.

However, the set of $\alpha$ satisfying {\bf A1} is substantially larger than the countable set of quadratic irrationals. In fact, this set has full Lebesgue measure in $(0,1)$, which can be seen by an elementary argument using the Borel-Cantelli Lemma (similar to how the same fact is proven for the set of non-Liouville numbers).

The condition {\bf A2} means that if there are several vectors in ${ \mathcal S}$
with the same direction, then they form a subset of a periodic
one-dimensional lattice. In \cite{KS} this is used as a technical condition in the analysis of certain one-dimensional operators associated with $H$. The condition {\bf A2}, just as {\bf A1}, can be considered as generically satisfied,  in the sense that it holds for typical finite symmetric subsets of $\R^2$, which don't have any vectors of the same direction other than what is needed for the symmetry requirement.

One way, but not the only way, in which {\bf A2} can be violated is for separable potentials. For example,
\begin{equation}
V(x_1,x_2) =\cos(2\pi x_1)+\cos(2\pi x_2)+\cos(2\pi \alpha x_1)+\cos(2\pi \alpha x_2)
\end{equation}
corresponds to ${\mathcal S} = \{(\pm 1,0), (0,\pm 1), (\pm \alpha, 0), (0, \pm \alpha)$, which does not satisfy {\bf A2}. Of course, separable potentials can be studied with the much more developed theory of one-dimensional quasi-periodic Schr\"odinger operators. We note that results on absolute continuity for 1D quasi-periodic operators, e.g.\ \cite{DS,El,JM,MP,R}, also typically require diophantine (non-Liouville) properties of the frequencies.

A basic example of a quasi-periodic potential satisfying all our assumptions and not being periodic in any direction is
\begin{equation}
V(x_1,x_2) =\cos(2\pi x_1)+\cos(2\pi x_2)+\cos(2\pi(\alpha x_1+x_2))+\cos(2\pi(x_1+\alpha x_2)),
\end{equation}
with $\alpha$ satisfying {\bf A1}. Here ${\mathcal S} = \{(\pm 1,0), (0,\pm 1), (\pm \alpha, \pm 1), (\pm 1, \pm \alpha)\}$, so {\bf A2} holds as well.

To conclude, we mention that \cite{KS} states a longer list of assumptions on the quasi-periodic potential $V$, some being consequences of the others, mostly to have these facts available for the proofs. The assumptions {\bf A1} and {\bf A2} are chosen minimal to imply all that is needed.

\subsection{The Main Result}

Now we consider \eqref{IVP}, $V$ being limit-periodic or quasi-periodic with assumptions as above. Clearly, the ballistic upper bound of \cite{RaSi} applies and we have \eqref{genballistic} for initial conditions $\Psi_0$ satisfying (\ref{RScond}).

%%%%\footnote{Question 1: Since we are considering $\R^2$, this multiplication operator $X$ is not really defined on the full $L^2(\R^2)$. If we define $X$ on a dense subset, %%%%%%say $C_c$, first, we cannot really extend $X$ continuously since $X$ is not bounded. What can we do?-Consider the {\bf closure} of the operator, see Akhiezer, %%%%%%%%Glazman}

As our main result we prove that under the above assumptions, both in the limit-periodic and quasi-periodic case, one also has corresponding {\it ballistic lower bounds} for a large class of initial conditions. For this we use the infinite-dimensional spectral projector $E_{\infty}$ for $H$ whose construction is described in Section~\ref{SpecPropH} below.

\begin{thm}\label{Thm1}
There is an infinite-dimensional projector $E_{\infty}$ in $L^2(\R^2)$ (described in Section~\ref{SpecPropH}) with the following property: For any
\begin{equation} \label{nontriv}
\Psi _0\in {\mathcal C_0^\infty} \quad \mbox{with} \quad E_{\infty }\Psi _0 \neq 0
\end{equation}
there are constants $c_1 = c_1(\Psi_0)>0$ and $T_0 = T_0(\Psi)$ such that the solution $\Psi (\vec x,t)$ of \eqref{IVP} satisfies the estimate
    \begin{equation}\label{ball-1}
    \frac{2}{T}\int _0^\infty e^{-2t/T} \big\|X \Psi (\cdot,t)\big \|^2_{L^2(\R^2)} dt >c_1 T^2
    \end{equation}
for all $T>T_0$.
\end{thm}

In Section 2 we show that $E_{\infty}$  is close in norm to ${\mathcal F}^*\chi \left(\mathcal{G}_{\infty}\right){\mathcal F}$, where ${\mathcal F}$ is the Fourier transform and $\chi \left(\mathcal{G}_{\infty}\right) $ the characteristic function of a set $\mathcal{G}_{\infty}$, which has asymptotically full measure in $\R^2$, see (\ref{full}) and (\ref{March21a}).

As already remarked in Section~\ref{prior}, due to the validity of the ballistic upper bound (\ref{genballistic}) for all initial conditions $\Psi_0 \in C_0^{\infty} \subset \mathcal{S}_1$, Theorem~\ref{Thm1} remains true if the Abel means are replaced by Cesaro means.

Theorem~\ref{Thm1} will be proven in two steps. First we will show

\begin{prop}\label{Prop2}
For $E_\infty$ as above, if $\Psi _0\in E_{\infty }\mathcal C_0^\infty$, $\Psi _0\neq 0$,
then the solution $\Psi (\vec x,t)$ of \eqref{IVP} satisfies the ballistic lower bound (\ref{ball-1}).
\end{prop}

Note that Proposition~\ref{Prop2} differs from Theorem~\ref{Thm1} by the fact that the initial condition $\Psi_0$ for which the ballistic lower bound is concluded is in the image of $C_0^{\infty}$ under the projection $E_{\infty}$ (but that $\Psi_0$ itself is not in $C_0^\infty$ here). This proposition takes the role of our core technical result, i.e.\ most of the technical work towards proving Theorem~\ref{Thm1} will go into the proof of the proposition. Theorem~\ref{Thm1} gives a more explicit description of initial conditions for which ballistic transport can be established. In fact, one easily combines Theorem~\ref{Thm1} with the ballistic upper bound (\ref{genballistic}) to get ballistic transport in form of a two-sided bound for many initial conditions:

\begin{cor} \label{Cor}
There is an $L^2$-dense and relatively open subset $\mathcal D$ of $C_0^\infty$ such that for every $\Psi_0 \in {\mathcal D}$ there are constants $0<c_1 \le C_1 <\infty$ such that the ballistic upper bound (\ref{genballistic}) and the ballistic lower bound (\ref{ball-1}) hold.
\end{cor}

This follows by an elementary argument using only that $E_{\infty}$ is not the zero projection and $C_0^{\infty}$ is dense in $L^2$ (and that $C_0^{\infty}$ functions also satisfy (\ref{RScond})).

It is certainly desirable to go beyond this corollary and to more explicitly characterize classes of initial conditions for which (\ref{nontriv}) holds. This requires to much better describe and exploit the nature of the projector $E_{\infty}$.  While we believe that $E_{\infty }\Psi _0 \neq 0$ for any $0\not=\Psi _0\in {\mathcal C_0^\infty}$, we do not have a proof of this. We will return to this question at the end of Section~\ref{Thm1proof}, see Remark~\ref{smoothEFE}, where we will more explicitly construct initial conditions which lead to both upper and lower ballistic transport bounds. These will have the form of suitably regularized generalized eigenfunction expansions.

The proofs of Theorem~\ref{Thm1} for the limit-periodic case \eqref{V} and the quasi-periodic case \eqref{Vq} are analogous, using results from \cite{KL3} and \cite{KS,KS1}, respectively. For the sake of definiteness, we  present all the details for the limit-periodic case only. Where necessary, we will comment on the quasi-periodic setting. Also, some statements in the paper \cite{KL3} on the limit-period case were presented in a form inconvenient for what we need here and thus need some technical adjustments, while the corresponding results for the quasi-periodic  case in \cite{KS1} are more directly applicable. This provides another reason for mostly focusing on the limit-periodic case here.

In Section~\ref{SpecPropH} we start by recalling results on the spectral properties of the operator $H$ which were obtained in \cite{KL3} for the limit-periodic case and in \cite{KS1} for the quasi-periodic case. Some of these results will also be adjusted and refined to make them more suitable for the proof of our main result.

Proposition~\ref{Prop2} will be proven in Section~\ref{Prop2proof}, followed by the proof of Theorem~\ref{Thm1} in Section~\ref{Thm1proof}. We conclude Section~\ref{Thm1proof} with a discussion of how to weaken the $C_0^\infty$ assumption in our results, arguing that it is enough to require a sufficient amount of smoothness and decay.  This will also shed some more light on the above question of describing initial conditions which lead to both {\it upper and lower} ballistic bounds.

In Section~\ref{appendices} we collect several appendices which provide technical details for some of the arguments from earlier sections.

\subsection{Comments on techniques and limitations}

To motivate the technical constructions used in this work, we conclude this introduction by a discussion of some of the main features of the methods used here and in the underlying works \cite{KL3,KS,KS1}. This will also shed light on the limitations of what can be done with our approach.

We start by recalling that the Fourier transform yields an explicit spectral resolution $E_{\lambda}$, $\lambda \in \R$ for $H_0=-\Delta$, given by
\begin{equation} E_{\lambda} F = \frac{1}{4\pi^2} \int_{\mathcal{G}_\lambda} (F, \Psi(\vec k,\cdot)) \Psi(\vec k, \cdot)\,d\vec k,\label{March21}
\end{equation}
where, say, $F$ is continuous and compactly supported, $(\cdot,\cdot)$ is the standard $L^2$ scalar product, ${\mathcal G}_\lambda:=\{\vec k\in \R^2: |\vec k|^2 \le \lambda\}$, and $\Psi(\vec k, \vec x) = e^{i \langle \vec k, \vec x \rangle}$ are the plane waves, i.e.\ generalized eigenfunctions with $H_0 \psi(\vec k,\cdot) = |\vec k|^2 \psi(\vec k, \cdot)$.

The central result in the works \cite{KL3,KS,KS1} is to find a spectral sub-resolution for the limit-periodic and quasi-periodic operator $H$, under the assumptions given above, which at high energy is asymptotically close to the Fourier transform. It takes the form
\[ E_{\infty,\lambda} F = \frac{1}{4\pi^2} \int_{\mathcal{G}_{\infty,\lambda}} (F, \Psi_{\infty}(\vec k,\cdot)) \Psi_{\infty}(\vec k,\cdot)\,d\vec k,\]
where $\Psi_{\infty}(\vec k, \cdot)$ are generalized eigenfunctions of $H$ close to the plane waves, with generalized eigenvalues $\lambda_{\infty}(\vec k)$ close to $|\vec k|^2$, as specified in (\ref{aplane}) to (\ref{16a}) below. The momentum integration is restricted to ${\mathcal G}_{\infty,\lambda} = \{ \vec k \in {\mathcal G}_{\infty}: \lambda_{\infty}(\vec k) \le \lambda\}$. Here the set ${\mathcal G}_{\infty}$ is a subset of $\R^2$ of asymptotically full measure, see (\ref{full}).  It follows (Section 2) that $E_{\infty,\lambda}$ is an operator close in norm to (\ref{March21}) with $\mathcal{G}_{\infty,\lambda}$ instead of just $\mathcal{G}_\lambda$. In other words, $E_{\infty,\lambda}$ is close to ${\mathcal F}^*\chi \left(\mathcal{G}_{\infty,\lambda}\right){\mathcal F}$, where $\chi \left(\mathcal{G}_{\infty,\lambda}\right)$ is the characterictic function of $\mathcal{G}_{\infty,\lambda}$. It follows that the operator  $E_{\infty}$ in Theorem 1.1 is close to ${\mathcal F}^*\chi \left(\mathcal{G}_{\infty}\right){\mathcal F}$.

By construction $\mathcal{G}_{\infty,\lambda}$ is a Cantor-type set: One considers a sequence of approximants $H_n$ of $H$ and, iteratively for each approximant, has to remove sets of resonant quasi-momenta from $\R^2$. The resonant sets are due to resonant eigenvalue splittings for $H_n$, regions where the perturbation theoretic arguments of \cite{KL3,KS,KS1} fail to work.

In the limit-periodic case one can choose $H_n = -\Delta +V_n$ with periodic approximants $V_n$ of $V$, see (\ref{W_n}). A substantially different approach is needed for the quasi-periodic case, where the approximants arise as part of a multiscale analysis procedure which is carried out by working in momentum space, see \cite{KS}. In particular, the $H_n$ are {\it not} found via periodic approximants $V_n$ of $V$ (e.g.\ via the continued fraction approximations $\alpha_n$ of $\alpha$), as the latter do not converge uniformly in the space variable. In fact, the condition {\bf A1} means that $\alpha$ is not well approximated by rationals. As noted above, this situation is familiar from results for the 1D case.

The $E_{\infty,\lambda}$ converge strongly as $\lambda\to\infty$ to the orthogonal projection $E_{\infty}$ appearing in our main results. The main difficulty in analyzing the corresponding branch of the spectrum of $H$ lies in the Cantor structure of ${\mathcal G}_{\infty}$. Arguments involving integration by parts and stationary phase techniques have to be carefully justified. The works \cite{KL3,KS,KS1} succeed in establishing absolute continuity of the corresponding spectral measures. However, more work and further analysis is required to obtain the ballistic transport bounds which are our goal here. In particular, this will require that in Section~\ref{SpecPropH} we recall the details of the construction of the functions $\Psi_{\infty}$, $\lambda_{\infty}$ and the set ${\mathcal G}_{\infty}$ from earlier works. A new tool to be exploited here is that the functions $\lambda_{\infty}(\vec k)$ and $\Psi_{\infty}(\vec k, \vec x)$ can be extended smoothly from $\vec k \in {\mathcal G}_{\infty}$ to $\vec k \in \R^2$ (the extensions are not eigenvalues and eigenfunctions anymore), see Sections~\ref{extlambda} and \ref{extpsi}.

A limitation of our perturbation theoretic approach is that it does not exclude the possibility of singular spectrum imbedded in the absolutely continuous spectrum at high energy. Also, our methods don't apply at energies near the bottom of the spectrum where Anderson localization might be expected (which should require methods similar to those for the lattice case discussed in the next paragraph). We see a chance that our methods can be extended to dimensions $d>2$, but this will be technically more difficult, as the number of resonant energies grows with the dimension. Another open problem in the multidimensional setting is to understand and analyze possible anomalous transport, i.e.\ situations with diffusion exponent $\beta$ strictly between 0 and 1 (where the dependence of $\beta$ on the moment $m$ in (\ref{moments}) becomes more significant). Finally, there still do not seem to be any mathematical methods to characterize the spectral type of self-similar potentials which model important physical examples of multidimensional {\it quasi-crystals} (while examples such as the ones considered here have been referred to as {\it modulated crystals}).

While the theory of one-dimensional limit and quasi-periodic potentials is very highly developed after more than three decades of research, the limitations mentioned above show that completing a similar program in the multi-dimensional case still seems far from reach. Our work is, in some sense, complementary to the results obtained in \cite{ChD, BGS, B2007}. These works establish Anderson localization for the multidimensional discrete Schr\"odinger equation with suitable quasi-periodic potential at {\it large coupling} (for dimension $d\ge 2$ in \cite{ChD} and, for more general quasi-periodic dynamics, in \cite{BGS} for $d=2$ and \cite{B2007} for $d\ge 2$). The high energy regime for continuum Schr\"odinger operators which we consider here has no direct analogue in the lattice case. One might expect that methods similar to ours can be used to show the existence of a.c.\ spectrum and ballistic transport for discrete 2D limit or quasi-periodic Schr\"odinger operators at {\it small coupling}, but we have not verified this. Observing that the results on Anderson localization in \cite{BGS,B2007} use a multi-scale analysis approach in position space, while our result is based on an MSA technique in momentum space, it is tempting to think that a multidimensional analogue of Aubry duality is emerging.

%%%%%%%%%%%%%%%%%%%%%%%%%%%%%%%%%%%%%%%%%%%%%%%%%%%%%%%%%%%%%%%%%%%%%%%%%%%%%%%%%%%%%%%%%%%%

\section{Spectral Properties of the Operator $H$} \label{SpecPropH}

Our proofs of Proposition~\ref{Prop2} and Theorem~\ref{Thm1} are based on the results and properties of two-dimensional limit and quasi-periodic Schr\"odinger operators derived in the papers \cite{KL3} and \cite{KS1}. While those works derived, in particular, the existence of an absolutely continuous component of the spectrum, we will show here how the bounds obtained can be used and, in part, improved, to also conclude ballistic transport. In this section we give a thorough discussion of the results and methods from \cite{KL3} and \cite{KS1}, mostly focusing on the limit-periodic case. In particular, we give a detailed construction of the spectral projection $E_{\infty}$ used in our main results.

\subsection{The Case of a Limit-Periodic Potential}

\subsubsection{Prior results}

To describe $E_{\infty}$, we recall
 the spectral properties of $H$, obtained in \cite{KL3}:
    \begin{enumerate}
    \item The spectrum of the operator (\ref{limper}), (\ref{V})
    contains a semiaxis. A proof of an analogous result by different means can be found
    in the  paper \cite{SkSo}. In \cite{SkSo}, the authors consider the
    operator $H=(-\Delta)^l+V$, $8l>d+3$, $d\neq 1(\mbox{mod}\; 4)$.
    This obviously includes our case $l=1$, $d=2$.
    However, there is an additional rather strong restriction on the potential
    $V(\vec x)$ in \cite{SkSo}, which we don't have here: In \cite{SkSo} all the period lattices of the potentials $V_r$ need to have a nonzero vector $\gamma$ in common,
    i.e., $V(\vec x)$ is periodic in direction $\gamma $.

    \item There are generalized eigenfunctions $\Psi_{\infty }(\vec k, \vec x)$,
    corresponding to the semiaxis, which are    close to plane waves:
    For every $\vec k $ in an extensive subset $\mathcal{G} _{\infty }$ of $\R^2$ (in the sense of (\ref{full}) below) there is
    a solution $\Psi_{\infty }(\vec k, \vec x)$ of the  equation
    $H\Psi _{\infty }=\lambda _{\infty }\Psi _{\infty }$ which can be described by
    the formula:
        \begin{equation} \label{aplane}
        \Psi_{\infty }(\vec k, \vec x)
        =e^{i\langle \vec k, \vec x \rangle}\left(1+u_{\infty}(\vec k, \vec x)\right),
        \end{equation}
        \begin{equation}\label{aplane1}
        \|u_{\infty}\|_{L^\infty(\R^2)}=_{|\vec k| \to \infty}O(|\vec k|^{-\gamma _1}),\ \ \ \gamma _1>0,
        \end{equation}
    where $u_{\infty}(\vec k, \vec x)$ is a limit-periodic function, as the potential.
    The  eigenvalue $\lambda _{\infty }(\vec k)$ corresponding to
    $\Psi_{\infty }(\vec k, \vec x)$ is close to $|\vec k|^{2}$:
        \begin{equation}
        \lambda _{\infty }(\vec k)=_{|\vec k| \to \infty}|\vec k|^{2}+
        O(|\vec k|^{-\gamma _2}),\ \ \ \gamma _2>0. \label{16a}
        \end{equation}
    The ``non-resonant" set $\mathcal{G} _{\infty }$ of the vectors $\vec k$, for which (\ref{aplane}) to (\ref{16a}) hold, is an extensive Cantor type set:
    $\mathcal{G} _{\infty }=\cap _{n=1}^{\infty }\mathcal{G}_n$,
    where $\{\mathcal{G} _n\}_{n=1}^{\infty}$ is a decreasing sequence of sets in $\R^2$. Each $\mathcal{G} _n$ has a finite number of holes in each bounded region. More and more holes appears when $n$ increases, however holes added at each step are of smaller and smaller size.
    The set $\mathcal{G} _{\infty }$ satisfies the estimate:
       \begin{equation}\label{full}
       \frac{\left|\left(\mathcal{G} _{\infty}\cap B_R\right)\right|}{\left| B_R\right|}
       =_{ R\to \infty }1+O( R^{-\gamma _3}),\quad \gamma_3 >0,
       \end{equation}
    where $B_R$ is the disk of radius $R$ centered at the origin, $|\cdot |$ is the Lebesgue measure in $\R^2$.

    \item The set $\mathcal{D}_{\infty}(\lambda)$, defined as a level (isoenergetic) set for
    $\lambda _{\infty }(\vec k)$,
    $$ {\mathcal D} _{\infty}(\lambda)=\left\{ \vec k \in \mathcal{G} _{\infty } :\lambda _{\infty }(\vec k)=\lambda\right\},$$ is proven to be a slightly distorted circle with an
    infinite number of holes. It can be described by  the formula:
        \begin{equation} \label{D}
        {\mathcal D}_{\infty}(\lambda)=\{\vec k:\vec k=\varkappa _{\infty}(\lambda, \vec{\nu})\vec{\nu},
        \ \vec{\nu} \in {\mathcal B}_{\infty}(\lambda)\},
        \end{equation}
    where ${\mathcal B}_{\infty }(\lambda )$ is a subset of the unit
    circle $S_1$. The set ${\mathcal B}_{\infty }(\lambda )$ can be
    interpreted as the set of possible directions of propagation for the
    almost plane waves (\ref{aplane}).  The set ${\mathcal B}_{\infty
    }(\lambda )$ has a Cantor type structure and an asymptotically full
    measure on $S_1$ as $\lambda \to \infty $:
        \begin{equation}\label{B}
        L\left({\mathcal B}_{\infty}(\lambda )\right)
        =_{\lambda \to \infty}2\pi +O\left(\lambda^{-\gamma _3/2}\right),
        \end{equation}
    here and below $L(\cdot)$ is Lebesgue measure on $S_1$.
    The value $\varkappa _{\infty }(\lambda ,\vec \nu )$ in (\ref{D}) is
    the ``radius" of ${\mathcal D}_{\infty}(\lambda)$ in a direction
    $\vec \nu $. The function $\varkappa _{\infty }(\lambda ,\vec \nu)-\lambda^{1/2}$ describes the deviation of ${\mathcal D}_{\infty}(\lambda)$ from the perfect circle of radius
    $\lambda^{1/2}$. It is proven that the deviation is asymptotically
    small, uniformly in $\vec \nu \in {\mathcal B}_{\infty}(\lambda)$:
        \begin{equation}\label{h}
        \varkappa _{\infty }(\lambda ,\vec \nu)
        =_{\lambda \to \infty} \lambda^{1/2}+O\left(\lambda^{-\gamma _4}\right), \ \ \ \gamma _4>0.
        \end{equation}

    \item Absolute continuity of the branch of the spectrum (the semiaxis) corresponding to $\Psi_{\infty }(\vec k, \vec x)$ is proven, see details below.

\end{enumerate}

\subsubsection{Description of methods:} \label{methods} To prove the above results  in \cite{KL3}, the authors considered the sequence of operators:
    $$ H_0=-\Delta , \ \ \ \ \ \
H_n=H_0+\sum_{r=1}^{M_n} V_r,\ \ \ n\geq 1, \ M_n \to \infty
\mbox{ as } n \to \infty .$$
Obviously, $\|H-H_n\|\to 0$ as $n\to \infty $, where $\|\cdot \|$ is the norm in the class of
bounded operators. Clearly,
    \begin{equation}
    H_n=H_{n-1}+W_n, \ \
    W_n=\sum_{r=M_{n-1}+1}^{M_n} V_r. \label{W_n}
    \end{equation}
Each operator $H_n$, $n\geq 1$, is considered as a perturbation of
the previous operator $H_{n-1}$. Every operator $H_n$ is
periodic, however the periods go to infinity as $n \to \infty$. It is shown
that there is a $\lambda_*$, $\lambda_*=\lambda_*(V)$, such
that the semiaxis $[\lambda _*, \infty )$ is contained in the
spectra of {\bf all} operators $H_n$. For every operator
$H_n$ there is a set of eigenfunctions (corresponding to the
semiaxis) being close to plane waves:
for every $\vec k $ in an extensive subset $\mathcal{G} _n$ of $\R^2$, there
is a solution $\Psi_{n}(\vec k, \vec x)$ of the differential equation
$H_n \Psi _n=\lambda _n\Psi _n$, which can be described by the formula:
    \begin{equation}\label{na}
    \Psi_n (\vec k, \vec x)
    =e^{i\langle \vec k, \vec x \rangle}\left(1+ u_{n}(\vec k, \vec
    x)\right),\ \ \
    \| u_{n}\|_{L^{\infty }(\R^2)}\underset{|\vec k| \to \infty}{=}O(|\vec k|^{-\gamma _1}),\ \ \ \gamma _1>0,
    \end{equation}
where $ u_{n}(\vec k, \cdot)$ has periods $2^{M_n-1}\vec d_1,
2^{M_n-1}\vec d_2$.
The corresponding eigenvalue $\lambda_n (\vec k)$ is close to $|\vec
k|^{2}$:
    \begin{equation}\label{back}
    \lambda_n(\vec k)=_{|\vec k| \to \infty}|\vec k|^{2}+
    O\left(|\vec k|^{-\gamma _2}\right),\ \ \ \gamma _2>0.
    \end{equation}
The non-resonant set $\mathcal{G} _{n}$ for which (\ref{back}) holds,
is proven to be extensive in $\R^2$:
    \begin{equation}\label{16b}
    \frac{\left|\mathcal{G} _{n}\cap
    B_R \right|}{\left| B_R\right|}=_{R\to \infty }1+O(R^{-\gamma _3}).
    \end{equation}
The estimates (\ref{na}) -- (\ref{16b}) are uniform in $n$.
The set ${\mathcal D}_{n}(\lambda)$ is defined as the level
(isoenergetic) set for the non-resonant eigenvalue $\lambda_n(\vec
k)$: $$ {\mathcal D} _{n}(\lambda)=\left\{ \vec k \in \mathcal{G}
_n:\lambda_n(\vec k)=\lambda \right\}.$$ This set is proven to
be a slightly distorted circle with a finite number of holes.
The set ${\mathcal D} _{n}(\lambda)$ can be
described by the formula:
     \begin{equation} {\mathcal D}_{n}(\lambda)=\{\vec k:\vec k=
    \varkappa_{n}(\lambda, \vec{\nu})\vec{\nu},
    \ \vec{\nu} \in {\mathcal B}_{n}(\lambda)\}, \label{Dn}
    \end{equation}
where ${\mathcal B}_{n}(\lambda )$ is a subset  of the unit circle
$S_1$. The set ${\mathcal B}_{n}(\lambda )$ can be interpreted as
the set of possible directions of propagation for  almost plane
waves (\ref{na}). It is shown that $\{{\mathcal
B}_n(\lambda)\}_{n=1}^{\infty }$  is a decreasing sequence of sets,
since on each step more and more directions are excluded.  Each
${\mathcal B}_{n}(\lambda )$ has an asymptotically full measure on
$S_1$ as $\lambda \to \infty $:
    \begin{equation}\label{Bn}
    L\left({\mathcalB}_{n}(\lambda )\right)=_{\lambda \to \infty }2\pi
    +O\left(\lambda^{-\gamma _3 /2}\right),
    \end{equation}
the estimate being uniform in $n$. The set ${\mathcal B}_{n}(\lambda)$ has
only a finite number of holes, however their number is growing with
$n$. More and more holes of a smaller and smaller size are added at
each step. The value $\varkappa_{n}(\lambda ,\vec \nu
)-\lambda^{1/2}$ gives the deviation of ${\mathcal D}_{n}(\lambda)$
from the perfect circle of  radius $\lambda^{1/2}$  in
direction $\vec \nu $. It is proven that the deviation is
asymptotically small uniformly in $n$:
    \begin{equation}\label{hn}
    \varkappa_{n}(\lambda ,\vec \nu)
    =\lambda^{1/2}+O\left(\lambda^{-\gamma _4 }\right),\ \ \ \
    \frac{\partial \varkappa_{n}(\lambda ,\vec \nu)}{\partial \varphi
    }=O\left(\lambda^{-\gamma _5 }\right), \ \ \ \gamma _4, \gamma _5>0,
    \end{equation}
$\varphi $ being an angle variable $\vec\nu =(\cos \varphi ,\sin \varphi )$.

On each step  more and more points are excluded from the
non-resonant sets $\mathcal{G} _n$ and, thus, $\{ \mathcal{G} _n
\}_{n=1}^{\infty }$ is a decreasing sequence of sets. The set
$\mathcal{G} _\infty $ is defined as the limit set: $\mathcal{G}
_\infty=\cap _{n=1}^{\infty }\mathcal{G} _n $. It has an infinite
number of holes in each bounded region, but nevertheless satisfies
the relation (\ref{full}). For every $\vec k \in \mathcal{G} _\infty
$ and every $n$, there is a generalized eigenfunction of $H_n$
of the type  (\ref{na}). It is proven that  the sequence of $\Psi
_n(\vec k, \vec x)$ has a limit in $L^{\infty }(\R^2)$ as $n\to
\infty$, when $\vec k \in \mathcal{G} _\infty $. The function $\Psi
_{\infty }(\vec k, \vec x) =\lim _{n\to \infty }\Psi _n(\vec k, \vec
x)$ is a generalized eigenfunction of $H$. It can be written in the
form (\ref{aplane})--(\ref{aplane1}). Naturally, the corresponding
eigenvalue $\lambda _{\infty }(\vec k) $ is the limit of $\lambda_n(\vec k )$ as $n \to \infty $.
Expansion with respect to the generalized eigenfunctions $\Psi
_{\infty }(\vec k, \cdot)$, $\vec k \in \mathcal{G} _\infty
$, will give a reducing subspace for $H$, with corresponding spectral resolution arising as the limit of spectral resolutions for the approximating periodic operators $H_n$.

To study these, one needs properties of the limit ${\mathcal B}_{\infty}(\lambda)$ of ${\mathcal
B}_n(\lambda)$:
    $${\mathcal B}_{\infty}(\lambda)=\bigcap_{n=1}^{\infty} {\mathcal
B}_n(\lambda),\ \ \ {\mathcal B}_n(\lambda) \subset {\mathcal B}_{n-1}(\lambda).$$
This set has a Cantor type structure on the unit circle. That it has asymptotically full measure (\ref{B}) follows from (\ref{Bn}).
We prove that the sequence $\varkappa _n(\lambda ,\vec \nu )$, $n=1,2,... ,$,
describing the isoenergetic curves $\mathcal{D}_n(\lambda)$,  quickly converges as $n\to \infty$.
Hence, ${\mathcal D}_{\infty}(\lambda)$ can be described as the
limit of ${\mathcal D}_n(\lambda)$ in the sense (\ref{D}), where
$\varkappa _{\infty}(\lambda, \vec{\nu})=\lim _{n \to \infty}
\varkappa _n(\lambda, \vec{\nu})$ for every $\vec{\nu} \in {\mathcal
B}_{\infty}(\lambda)$. It is shown that the derivatives of the
functions $\varkappa _n(\lambda, \vec{\nu})$ (with respect to the
angle variable $\varphi $ on the unit circle) have a limit as $n\to
\infty $ for every $\vec{\nu} \in {\mathcal B}_{\infty}(\lambda)$.
We denote this limit by $\frac{\partial \varkappa_{\infty}(\lambda
,\vec \nu)}{\partial \varphi }$. Using (\ref{hn}) we  prove that
    \begin{equation}\label{Dec9a}
    \frac{\partial \varkappa_{\infty}(\lambda ,\vec \nu)}{\partial
    \varphi }=O\left(\lambda^{-\gamma _5 }\right).
    \end{equation}
Thus, the limit curve ${\mathcalD}_{\infty}(\lambda)$ has a tangent vector in spite of its Cantor
type structure, the tangent vector being the limit of the
corresponding tangent vectors for ${\mathcal D}_n(\lambda)$ as $n\to
\infty $.  The curve ${\mathcal D}_{\infty}(\lambda)$ takes the form of
  a slightly distorted circle with
an infinite number of holes.

Let   $\mathcal{G}_n'$ be a bounded Lebesgue measurable subset of  $\mathcal{G}_n$.
We consider the spectral projection
$E_n\left(\mathcal{G}_n'\right)$ of $H_n$, corresponding to functions
$\Psi _n(\vec k ,\vec x)$, $\vec k \in \mathcal{G}_n'.$
By \cite{Ge}, $E_n\left( \mathcal{G}'_n\right): L^2(\R^2)\to
L^2(\R^2)$ can be represented by the formula:
    \begin{equation}\label{s}
    E_n\left( \mathcal{G}'_n\right)F=\frac{1}{4\pi ^2}\int
    _{ \mathcal{G}'_n}\bigl( F,\Psi _n(\vec
    k )\bigr) \Psi _n(\vec
    k) \,d\vec k
    \end{equation}
for any $F\in {\mathcal C_c}(\R^2)$, the continuous, compactly supported functions on $\R^2$. Here and below $\bigl( \cdot ,\cdot \bigr)$
is the canonical scalar product in $L^2(\R^2)$, i.e.,
    $$
    \bigl( F,\Psi _n(\vec k )\bigr)=\int _{\R^2}F(x)\overline{\Psi _n(\vec k ,\vec x)}\,d\vec x.
    $$
The above formula can be rewritten in the form
    \begin{equation}\label{ST}
    E_n\left(\mathcal{G}'_n\right)=S_n\left(\mathcal{G}'_n\right)T_n \left(
    \mathcal{G}'_n\right),
    \end{equation}
    $$T_n:{\mathcal C_c}(\R^2) \to L^2\left(  \mathcal{G}'_n\right), \ \
    \ \ S_n:L^\infty\left( \mathcal{G}'_n\right)\to L^2(\R^2),$$
    \begin{equation}\label{eq2}
    (T_nF)(\vec k) =\frac{1}{2\pi }\bigl( F,\Psi _n(\vec
    k )\bigr) \mbox{\ \ for any $F\in {\mathcal C_c}(\R^2)$},
    \end{equation}
$T_nF$ being in $L^{\infty }\left(  \mathcal{G}'_n\right)$, and
    \begin{equation}\label{ev}
    (S_nf)(\vec x) = \frac{1}{2\pi }\int _{  \mathcal{G}'_n}f (\vec k)\Psi _n(\vec
    k ,\vec x)\,d\vec k  \mbox{\ \ for any $f \in L^{\infty }\left(
    \mathcal{G}'_n\right)$.}
    \end{equation}
By \cite{Ge},
 \beq\label{eq:T_n bound}
 \|T_nF\|_{L^2\left( \mathcal{G}'_n\right)}\leq \|F\|_{L^2(\R^2)}
 \eeq
and
 \beq\label{eq:S_n bound}
 \|S_nf \|_{L^2(\R^2)}\leq \|f \|_{L^2\left(\mathcal{G}'_n\right)}.
 \eeq
Hence, $T_n$ and $S_n$ can be extended by continuity from ${\mathcal C_c}(\R^2)$ and $L^{\infty }\left(  \mathcal{G}'_n\right)$ to $L^2(\R^2)$ and $L^2\left( \mathcal{G}'_n\right)$,
respectively. Obviously, $T_n^{*}=S_n$. Thus, the operator $E_n\left( \mathcal{G}'_n\right)$ is described by (\ref{ST}) in the whole space $L^2(\R^2)$.

In what follows we will use these operators for the case where $\mathcal{G}'_n$ is given by
\begin{equation}
\mathcal{G}_{n, \lambda}=\{ \vec k \in {\mathcal{G}}_n:
\lambda_n(\vec k) < \lambda\}. \label{d}
\end{equation}
for finite sufficiently large $\lambda$. This set is Lebesgue measurable since ${\mathcal{G}}_n $ is
open and $\lambda_n(\vec k)$ is continuous on $
{\mathcal{G}}_n$.

Let
    \begin{equation}\label{dd}
    \mathcal{G}_{\infty, \lambda }=\left\{\vec k \in
    \mathcal{G}_{\infty }: \lambda _{\infty }(\vec k )<\lambda
    \right\}.
    \end{equation}
The function $\lambda _{\infty }(\vec k )$ is a Lebesgue
measurable function, since it is a limit of a sequence of
measurable functions. Hence, the set  $\mathcal{G}_{\infty, \lambda
}$ is measurable. The sets $\mathcal{G}_{n, \lambda
}$ and $\mathcal{G}_{\infty, \lambda
}$ are also bounded. It is shown in \cite{KL3} that
the measure of the symmetric difference of the
two sets $\mathcal{G}_{\infty, \lambda }$ and $\mathcal{G}_{n,
\lambda}$ converges
 to zero as $n \to
\infty$, uniformly in $\lambda$ in every bounded interval:
    $$\lim _{n\to \infty }\left|\mathcal{G}_{\infty, \lambda }\,\Delta\, \mathcal{G}_{n, \lambda
    }\right|=0.$$

Next, we consider the sequence of operators $T_n(\mathcal{G}_{\infty , \lambda})$ which are given by (\ref{eq2}) and act from $L^2(\R^2)$ to $L^2(\mathcal{G}_{\infty , \lambda})$.
It is proven in \cite{KL3} that the sequence $T_n(\mathcal{G}_{\infty ,\lambda})$  has a strong limit $T_{\infty}(\mathcal{G}_{\infty , \lambda})$. The operator $T_{\infty }(\mathcal{G}_{\infty , \lambda})$ satisfies $\|T_{\infty }\|\leq 1$ and can be described by the formula
$(T_{\infty }F)(\vec k)=\frac{1}{2\pi }\bigl( F,\Psi _{\infty }(\vec k)\bigr) $ for any $F\in {\mathcal C_c}(\R^2)$.
The convergence of $T_n(\mathcal{G}_{\infty , \lambda})F$ to
$T_{\infty }(\mathcal{G}_{\infty , \lambda})F$ is uniform in $\lambda $ for every $F\in L^2(\R^2)$.
We also consider the sequence of operators $S_n(\mathcal{G}_{\infty ,\lambda })$ which are given by (\ref{ev}) with ${\mathcal G}_n'=\mathcal{G}_{\infty , \lambda}$:
    \begin{equation} \label{SnGDef}
    S_n(\mathcal{G}_{\infty , \lambda}):\ L^2(\mathcal{G}_{\infty , \lambda})\to L^2(\R^2).
    \end{equation}
It is proven in \cite{KL3} that the sequence of operators
$S_n(\mathcal{G}_{\infty , \lambda})$ has a strong limit $S_{\infty }(\mathcal{G}_{\infty , \lambda})$. It follows $T^*_{\infty }(\mathcal{G}_{\infty , \lambda})=S_{\infty }(\mathcal{G}_{\infty , \lambda})$. Moreover, a slight modification of the proof (see Appendix 1 below) gives convergence  in operator norm sense as $n\to\infty$, uniform in $\lambda$.  Moreover, the estimate
 \begin{equation}
    \label{June5} \|S_{\infty }(\mathcal{G}_{\infty , \lambda})-S_0(\mathcal{G}_{\infty , \lambda})\|<c\lambda _*^{-\gamma _6}, \ \ \gamma _6>0,
    \end{equation}
    and, therefore,
     \begin{equation}
    \label{June5a} \|T_{\infty }(\mathcal{G}_{\infty , \lambda})-T_0(\mathcal{G}_{\infty , \lambda})\|<c\lambda _*^{-\gamma _6}, \ \ \gamma _6>0,
    \end{equation}
     holds for $\lambda > \lambda_*$, $c$ not depending on $\lambda $.

%%%%Corresponding result is proven in {KaSh} for quasi-periodic potential.
The operator $S_{\infty }(\mathcal{G}_{\infty , \lambda})$ satisfies $\|S_{\infty }\|=1$ and can be described by the formula
    \begin{equation}\label{ev1}
    (S_{\infty }f) (\vec x)= \frac{1}{2\pi }\int _{\mathcal{G}_{\infty , \lambda}}f (\vec k)\Psi _{\infty }(\vec
    k,\vec x) \,d\vec k
    \end{equation}
for any $f \in L^{\infty }\left( \mathcal{G}_{\infty , \lambda}\right)$.

The spectral projections
$E_n(\mathcal{G}_{\infty , \lambda })$ converge in norm to
$E_{\infty }(\mathcal{G}_{\infty , \lambda })$ in $L^2(\R^2)$ as $n$
tends to infinity, since $T_n=S_n^*$. The operator $E_{\infty }(\mathcal{G}_{\infty ,
\lambda })$ is a spectral projection of $H$. It can be represented
in the form $E_{\infty }(\mathcal{G}_{\infty , \lambda })=S_\infty(\mathcal{G}_{\infty , \lambda })
T_{\infty }(\mathcal{G}_{\infty , \lambda })$, where $S_{\infty }(\mathcal{G}_{\infty , \lambda })$ and $T_{\infty }(\mathcal{G}_{\infty , \lambda })$ are limits in norm of $S_n(\mathcal{G}_{\infty , \lambda })$ and
$T_n(\mathcal{G}_{\infty , \lambda })$, respectively.

For any $F\in
{\mathcal C_c}(\R^2)$ we have
    \begin{equation}\label{s1}
    E_{\infty }\left(\mathcal{G}_{\infty , \lambda }\right)F=\frac{1}{4\pi ^2}\int
    _{ \mathcal{G}_{\infty , \lambda }}\bigl( F,\Psi _{\infty }(\vec
    k)\bigr) \Psi _{\infty }(\vec
    k ) \,d\vec k ,
    \end{equation}
    \begin{equation}\label{s1uu}
    HE_{\infty }\left(\mathcal{G}_{\infty , \lambda }\right)F=\frac{1}{4\pi ^2}\int
    _{ \mathcal{G}_{\infty , \lambda }}\lambda _{\infty }(\vec k )
    \bigl( F,\Psi _{\infty }(\vec k )\bigr) \Psi _{\infty }(\vec k ) \,d\vec k .
    \end{equation}

 One also has the Parseval formula
\begin{equation} \label{Parseval}
\| E_{\infty}(\mathcal{G}_{\infty,\lambda }) F\|^2 = \frac{1}{4\pi^2} \int_{\mathcal{G}_{\infty,\lambda}} |(F, \Psi_{\infty}(\vec k)|^2\,d \vec k
\end{equation}
and the estimate
 \begin{equation}
    \label{June5b} \|E_{\infty }(\mathcal{G}_{\infty , \lambda})-S_0T_0(\mathcal{G}_{\infty , \lambda})\|<c\lambda _*^{-\gamma _6}, \ \ \gamma _6>0.
    \end{equation}
    Note that
    \begin{equation}
    S_0T_0(\mathcal{G}_{\infty , \lambda})={\mathcal F}^*\chi (\mathcal{G}_{\infty , \lambda}){\mathcal F}. \label{June5c}
    \end{equation}

   %%%%%%% $\bigl( F,\Psi _{\infty }(\vec
 %%%%%%%   \varkappa )\bigr) \Psi _{\infty }$ being an integral analogous
 %%%%%%%%   to the dot product in $L^2(R^2)$:
 %%%%%%%%    \begin{equation}\bigl( F,\Psi _{\infty}(\vec
%%%%%%%%     \varkappa )\bigr)=\int _{R^2}F(x)\overline{\Psi _{\infty }(\vec
%%%%%%%%     \varkappa ,x)}dx. \label{eq1}
%%%%%%%%     \end{equation}

The projections $E_\infty
(\mathcal{G}_{\infty,\lambda })$ are increasing in $\lambda$ and have a strong limit
$E_\infty(\mathcal{G}_{\infty})$ as $\lambda $ goes to infinity.
Hence,  the operator $E_{\infty
}(\mathcal{G}_{\infty})$ is a projection. The projections
$E_\infty(\mathcal{G}_{\infty,\lambda })$, $\lambda \in \R$, and
$E_\infty(\mathcal{G}_{\infty})$ reduce the operator $H$.
The family of projections
$E_\infty(\mathcal{G}_{\infty,\lambda} )$ is the resolution of the
identity of the operator $HE_\infty(\mathcal{G}_{\infty})$ acting in
$E_\infty(\mathcal{G}_{\infty})L^2(\R^2)$.  Further we denote $E_\infty(\mathcal{G}_{\infty})$  just by $E_\infty $ and use
\begin{equation}
\|E_\infty  -{\mathcal F}^*\chi (\mathcal{G}_{\infty }){\mathcal F}\|<c\lambda _*^{-\gamma _6}, \ \ \gamma _6>0.\label{March21a}
    \end{equation}
Obviously, the r.h.s.\ can be made arbitrary small by an appropriate choice of $\mathcal{G}_{\infty }$.

Absolute continuity of the restriction of $H$ to the range of $E_{\infty}$ is established in \cite{KL3}. In addition to the above mentioned convergence of the spectral projections of $H_n$ to those of $H$, uniform in $\lambda \ge \lambda_*$ for sufficiently large $\lambda_* = \lambda_*(V)$, this requires an analysis of the continuity properties of  the level curves
$\mathcal{D}_{\infty }(\lambda )$ with respect to $\lambda $.

In what follows, we want to establish that this branch of the spectrum of $H$ also leads to ballistic transport. For this we may need to increase the parameter $\lambda_*$ in a controlled way. We denote the new $\lambda_*$ by $\lambda_{**}$, with value to be specified later. This means a  change of the set ${\mathcal G}_\infty $: Instead we will consider the set ${\mathcal G}_\infty \setminus B_{k^{**}}$, where here and below $k^{**} = \sqrt{\lambda^{**}}$. With a slight abuse of notation we will denote this set again by ${\mathcal G}_{\infty}$.
For any fixed value of $\lambda _{**}$ the  projector $E_\infty({\mathcal G}_\infty)$ corresponds to a sufficiently rich branch of the absolutely continuous spectrum covering the half-line $[\lambda _{**},\infty)$.

%%%%%%%%%%%%%%%%%%%%%%%%%%%%%%%%%%%%%%%%%%%%%%%%%%%%%%%%%%%%%%%%%%%%%%%%%%%%%%%%%%%%%%%%%%%%
%\section{Spectral Properties of Operator $H$ with a quasi-periodic potential}

%%%%%%%%%%%%%%%%%%%%%%%%%%%%%%%%%%%%%%%%%%%%%%%%%%%%%%%%%%%%%%%%%%%%%%%%%%%%%%%%%%%%%%%%%%%%%%%%%%

%%%%%%%%Let $\hatt \Psi_0 \in C^\infty$ have a compact support $K$, then
%%%%%Let K(k_0)=\{\vec k\in \R^2,\ k_0<|\vec k|<2k_0\}$, where $k_0\leq k_*$,

\subsubsection{Extension of $\lambda_\infty(\vec k)$  from $\mathcalG_\infty$  to $\R^2$} \label{extlambda}

First, we extend the function  $\lambda_\infty(\vec k)$ from $\mathcalG_\infty$ to $\R^2$, the result being a $C^M(\R^2)$ function.
Note that the extended function is not an eigenvalue outside of $\mathcalG_\infty$.
%Let $\vec u_{\infty }(k)\in L^2(\Z^3)$ be the vector of Fourier coefficients of $u_{\infty }(\vec k,x)$, see \eqref{aplane}, \eqref{aplane1}.

Indeed, let $M$ be a natural number (in fact, we will need $M=7$ later). First,
%We start with the case of a limit-periodic potential.
following \cite{KL3}, we represent $\lambda _{\infty }(\vec k)-k^2$, $k:=|\vec k|$, %and  $\vec u_{\infty }(k)$,
$\vec k \in {{\mathcal G}}_{\infty }$,
in the form:
$$\lambda _{\infty }(\vec k)-k^2=\lambda_1(\vec k)-k^2+\sum _{n=1}^{\infty }\left( \lambda_{n+1}(\vec k)-\lambda_n(\vec k)\right). $$
%$$\vec u_{\infty }(\vec k)=\vec u^{(1)}(\vec k)+\sum _{n=1}^{\infty }\left( \vec u^{(n+1)}(\vec k)-\vec u^{(n)}(\vec k)\right). $$
By Theorem 2.6 in \cite{KL3}, with $D^m := \partial_1^{m_1} \partial_2^{m_2}$ we obtain
 \begin{equation}\label{eigenvalue-1}
 \left|  D^m_{}\left(\lambda_1(\vec k)-k^2\right)\right|<Ck^{-\gamma _2+\gamma _0|m|}
 \end{equation}
when $\vec k $ is in the $k^{-\gamma _0}$-neighborhood of ${{\mathcal G}}_1\supset {{\mathcal G}}_{\infty }$ and the constant depends only on $V$ and $m$.
Moreover, by Theorem 3.8 and  e.t.c.\ in \cite{KL3},
 \begin{equation}\label{eigenvalue-n}
 \left|  \lambda_{n+1}(\vec k)-\lambda_n(\vec k)\right|<e^{-k^{\eta s_n}},
 \end{equation}
for any $n \geq 1$, where $s_n=2^{n-1}s_1,$ $s_1$ being chosen sufficiently small with $0<s_1<10^{-4}$.
The value of $s_1$  is chosen at the beginning of the iteration procedure and, eventually,  $\lambda_*(V)$ and the constants in the estimates  depend on $s_1$.
Estimate \eqref{eigenvalue-n} is valid in the $(\epsilon _n k^{-1-\delta_0 })$-neighborhood of each $\vec k \in {{\mathcal G}}_n\supset {{\mathcal G}}_{\infty }$, where
$\epsilon _n=e^{-\frac 14k^{\eta s_n}}$ and  $\delta_0 >0$.
The constant $\frac 14$ in the definition of $\epsilon _n$, see \cite{KL3}, is chosen at random.
Instead of $\frac 14 $, one can take any fraction $\frac{1}{M+1}$, $M\geq 1$. This will lead, generally speaking, to an increase of $\lambda_*(V)$, when $M>3$.
We will denote the new $\lambda _*(V)$ by $\lambda _{**}(V,M)$.
Further we use the notation $\epsilon _n=e^{-\frac{1}{M+1}k^{\eta s_n}}$ and assume $k^2>\lambda _{**}(V,M)$.
Then we can rewrite \eqref{eigenvalue-n} as
 \begin{equation}\label{eigenvalue-n*}
 \left| \lambda_{n+1}(\vec k)-\lambda_n(\vec k)\right|< \epsilon _n^{M+1}
 \end{equation}
in the $(\epsilon _n k^{-1-\delta_0 })$-neighborhood of any $\vec k\in {{\mathcal G}}_n$.
Using analyticity of $\lambda _{n+1}(\vec k)$ and $\lambda _{n}(\vec k)$ in the complex $(\epsilon _n k^{-1-\delta_0 })$-neighborhood of any $\vec k \in {{\mathcal G}}_n$,
we obtain (see \cite{KL3})
 \begin{equation}\label{derivative-eigenvalue-n*}
 \left| D^m\left(\lambda_{n+1}(\vec k)-\lambda_n(\vec k)\right) \right|< \epsilon _n^{M+1-|m|} k^{(1+\delta_0 )|m|}
 \end{equation}
in ${{\mathcal G}}_n$ for all $m$.
Next, let $\eta _1(\vec k)$ be a function in $C^\infty$ with support in the (real) $k^{-\gamma _0}$-neighborhood of ${{\mathcal G}}_1$,
satisfying $\eta _1=1$ on ${{\mathcal G}}_1$ and $\left|D^{m}\eta _1(\vec k)\right|<k^{\gamma _0|m|}$.
This is possible since we can take a convolution of the characteristic function of the $\frac 12 k^{-\gamma _0}$-neighborhood of ${{\mathcal G}}_1$ with  $\omega (2k^{\gamma _0} \vec k)$,
where $\omega $ is a smooth cut-off function with support in the unit disc centered at the origin.
Similarly,  let $\eta _n(\vec k)$, $n\geq 2$,  be a $C^\infty$ function with support in the $(\epsilon _n k^{-1-\delta_0 })$-neighborhood of ${{\mathcal G}}_n$,
satisfying $\eta _n=1$ on ${{\mathcal G}}_n$ and
 \begin{equation}\label{derivative-eta-n}
 \left|D^{m}\eta _n (\vec k)\right|\leq \left(\epsilon _n k^{-1-\delta_0 }\right)^{-|m|}.
 \end{equation}
In the estimate \eqref{eigenvalue-1}, $\gamma _2=2-30s_1-20\delta_0 $, $\gamma _0=1+16s_1+11\delta_0 $. However,  \eqref{eigenvalue-1} can be improved when
$|m|<k^{s_1/2}$, see Lemma 2.5 in \cite{KL3}.
In this case, one can take $\gamma _0=3s_1+2\delta_0 $.
Choose $s_1$ small enough so that $2\gamma_0 M < \gamma_2$, i.e.,
 \bdm
 2(3s_1+2\delta_0 )M<2-30s_1-20\delta_0
 \edm
and for sufficiently large $k$, $M < k^{s_1/2}$ and so \eqref{eigenvalue-1} holds with $\gamma _0=3s_1+2\delta_0$.

Next, we extend $\lambda _{\infty }(\vec k)-k^2$ from ${{\mathcal G}}_{\infty }$ to $\R^2$ using the formula
\begin{equation}\lambda _{\infty }(\vec k)-k^2=(\lambda_1(\vec k)-k^2)\eta _1(\vec k)+\sum _{n=1}^{\infty }\left( \lambda_{n+1}(\vec k)-\lambda_n(\vec k)\right)\eta _{n+1}(\vec k).\label{May31} \end{equation}
It follows from \eqref{eigenvalue-n*} and \eqref{derivative-eigenvalue-n*} that the series converges in $C^M(\R^2)$.
Moreover, the next lemma follows from \eqref{eigenvalue-n*}--\eqref{derivative-eta-n}.
% Formula \eqref{eigenvector} can be obtained in an analogous way.
\begin{lem}\label{lambda}
For every natural number $M$, there exists $\lambda_{**}(V,M)>0$ such that
 the function $\lambda _{\infty }(\vec k)-k^2$ can be extended, as a
${\mathcal C^{M}}$ function, from ${{\mathcal G}}_{\infty,\lambda_{**}}$ to $\R$ and it satisfies
\begin{equation}\label{eigenvalue}
\left|  D^m \left(\lambda _{\infty }(\vec k)-k^2\right)\right|<C_M k^{-\gamma _2+\gamma _0|m|},
\end{equation}
for any $ m \in \N_0^2$ with $|m| \leq M$,
where $-\gamma _2+2\gamma_0 M < 0$.
%\begin{equation}
%\left\| D^m_{}\vec u_{\infty }(\vec k)\right\|<C_Mk^{-\gamma _1+\gamma _3|m|}, \label{eigenvector}
%\end{equation}
%where $\gamma _i>0$, $\gamma _i=\gamma _i(M)$, $2\gamma _3M<\gamma _2$, and $\gamma _i-s$ do not depend on $V, k$.
\end{lem}
\begin{remark}
 For our needs $M=7$ is sufficient and in what follows we assume that the corresponding $s_1$ and  $\lambda_{**}$ are chosen for $M=7$.
 \end{remark}

 \subsubsection{Extension of $\Psi_{\infty }(\vec k,\vec x)$ from $\mathcalG_\infty$ to $\R^2$} \label{extpsi}

We extend $\Psi_\infty(\vec k,\vec x)$ by a formula analogous to \eqref{May31}:
%%%nd $S_\infty({\mathcalG}_\infty ')$ which converges in norm
%%%, for any subset $\wti{\mathcalG}_\infty '$ of $\wti{\mathcalG}_\infty$ similarly to $\lambda_\infty$ as follows.
\begin{equation} \label{May31a}
 \Psi_\infty(\vec k,\vec x)-e^{i\la \vec k, \vec x \ra}
 =\left(\Psi_1(\vec k,\vec x)-e^{i\la \vec k, \vec x \ra}\right)\eta_1(\vec k)
  + \sum_{n=1}^\infty \left(\Psi_{n+1}(\vec k,\vec x)-\Psi_n(\vec k,\vec x)\right)\eta_{n+1}(\vec k).\end{equation}
  The series converges by \eqref{s++app}. Using the last formula and \eqref{ev1}, we define $S_\infty(\wti{\mathcalG}_\infty )$ for any $\wti{\mathcalG}_\infty \supset {\mathcalG}_\infty $:

\begin{align} \label{Sinf} \left(S_\infty(\wti{\mathcalG}_\infty )f\right)(\vec x)
  &:= \frac{1}{2\pi}\int_{\wti{\mathcalG}_\infty } f(\vec k) \Psi_\infty(\vec k,\vec x)\,d\vec k.
   \end{align}
  It is easy to see that
 %%%%% \begin{equation}
 %%% S_\infty(\wti{\mathcalG}_\infty )f
%%  = \left(S_0(\wti{\mathcalG}_\infty )f \right)(\vec x)+\sum_{n=0}^\infty \left(\bigl(S_{n+1}(\wti{\mathcalG}_\infty )-S_n(\wti{\mathcalG}_\infty ) \bigr)\eta_{n+1}f \right)(\vec x),
%%\end{equation}
\begin{equation} \label{June}
  S_\infty(\wti{\mathcalG}_\infty )
  = S_0(\wti{\mathcalG}_\infty )+\sum_{n=0}^\infty \bigl(S_{n+1}(\wti{\mathcalG}_\infty )-S_n(\wti{\mathcalG}_\infty ) \bigr)\eta_{n+1},
\end{equation}
 where $S_0(\wti{\mathcalG}_\infty )$ is defined by
 $$S_0(\wti{\mathcalG}_\infty )f= \frac{1}{2\pi}\int _{\wti{\mathcalG}_\infty }f(\vec k)e^{-i\la \vec k, \vec x \ra}d\vec k,$$
 $\eta_{n+1}$ is multiplication by $\eta_{n+1}(\vec k)$ and $S_n(\wti{\mathcalG}_\infty ) $ is given by \eqref{ev} with ${\mathcalG}_n'$ being the intersection of
 $\wti{\mathcalG}_\infty $ with the $(\epsilon _n k^{-1-\delta_0 })$-neighborhood of ${{\mathcal G}}_n$ for $n\geq 2$ and  the $k^{-\gamma _0}$-neighborhood of ${{\mathcal G}}_1$ for $n=1$.

%%%%%Moreover, using Lemma \ref{Lem2}, we get
Similarly to \eqref{June5}, we show that
 \begin{equation}
  \|S_\infty(\wti{\mathcalG}_\infty )-S_0(\wti{\mathcalG}_\infty )\|<c(V)\lambda_{**}^{-\gamma _6}.
  %%%%%%\ \ \hbox{as}\ \ k_{**}\to\infty.
  \label{2.36}
  \end{equation}
In what follows we assume that $\lambda_{**}$ is chosen so that, in particular, $c(V)\lambda_{**}^{-\gamma _6}\leq 1/2$. Thus we have
 \begin{equation}\label{S_infty}
 \| S_\infty(\wti{\mathcalG}_\infty )\|\leq 2.
 \end{equation}
Similarly, with $T_0$ the Fourier transform,
\begin{align}\label{T_infty}
  (T_\infty F)(\vec k) & := \frac{1}{2\pi} (F(\cdot), \Psi_\infty(\vec k, \cdot)) \notag \\
  &= (T_0 F)(\vec k)+ \sum_{n=0}^\infty \bigl((T_{n+1}-T_n)F\bigr)(\vec k)\eta_{n+1}(\vec k).
  \end{align}
 .

%To prove Theorem \ref{Thm2}, we give a lemma first.
We need one more auxiliary result.

 \begin{lem}\label{Ttransform}
 For any given $L\in\N$ there exists $\lambda_{**}(V,L)$ such that for any $F \in {\mathcal C^\infty_0}(\R^2)$, the function $T_\infty F$ as defined above is in ${\mathcal C^L}(\R^2)$. Moreover, if $0\leq j \leq L$ and $m \in \N_0^2,\ |m|\leq L$, then
  \begin{equation}\label{T_infty bound}
  \left| |\vec k |^j D^m (T_\infty F) (\vec k)\right| < C(L, F),
  \end{equation}
 for all $\vec k \in \R^2$.
 \end{lem}
 A proof of this lemma is given in Appendix 2 below.
 \begin{remark}
%% More sophisticated consideration shows that $T_\infty F$ belongs to the Schwartz class, but we don't need it here.
In fact, for our needs $L=6$ is sufficient and in what follows we assume that the corresponding $\lambda_{**}$ is chosen for $L=6$.
 \end{remark}

 \subsection{The Case of a Quasi-periodic Potential.}
%%\begin{rem}
The main results in the case of quasi-periodic potential \cite{KS1} are completely analogous to those for limit-periodic potential in Section 2.1.1, the only difference being that $u_{\infty}$ in
\eqref{aplane} is quasi-periodic, i.e., has a representation analogous to that for the potential, but not necessarily a trigonometric polynomial. The operators $H_n$ in the approximation
procedure are, naturally, quite different from \eqref{W_n}. However, the rest of Section 2.2.2 is completely analogous for both types of potentials, the quasi-periodic case being even somewhat simpler, since convergence of the sequence $S_n$ in
norm,  proven in Appendix 1 for the limit-periodic case,  is already proven in   \cite{KS}, \cite{KS1} for the quasi-periodic potential.
 The extension of $\lambda_\infty(\vec k)$  and $E_\infty(\vec k)$ to $\R^2$ (Section 2.1.3) is also completely similar in both cases.
%%The proof in quasi-periodic setting is absolutely similar and based on the analogues of Lemma \ref{eigenvalue} and Remark \ref{raz}. Corresponding a priori estimates can be %%found in {KaSh}.
%%\end{rem}
Note only that in the quasi-periodic case  Lemma~\ref{lambda} holds with $\gamma _2=2-88\mu \delta$, $\gamma _0=(40\mu +1)\delta$, $\delta>0$, by Theorem 3.3, Corollary 3.4 and Lemma 3.5 in \cite{KS1} and $\epsilon _n^{M+1}=k^{-\frac{\beta }{10}k^{r_{n-1}-r_{n-2}}}$ (see \eqref{eigenvalue-n*}), here
$\beta$ is a positive constant, $r_n$ is an increasing sequence going to infinity as $n\to \infty $, see Corollaries 5.4, 6.4 in \cite{KS1}.
%%\end{rem}

\section{Proof of Proposition~\ref{Prop2}} \label{Prop2proof}

Let $\mathcalS$ be the class of functions in $T_\infty{\mathcal C_0^\infty}(\R^2)$, see \eqref{T_infty}. As shown in Lemma~\ref{Ttransform},
 if $\hatt{\Psi}_0 \in \mathcalS$, then
 \begin{equation}
 \bigl| |\vec k|^j D^{m}(\hatt{\Psi}_0)(\vec k)| <C(j,m,\hatt{\Psi}_0)
 \end{equation}
 for any $\vec k \in \R^2$ when $j \leq 6$ and $|m| \leq 4$.

 Let $\hatt{\Psi}_0 \in \mathcalS$ and
\beq\label{Def:Psi}
\Psi(\vec x,t):= \frac{1}{2\pi}\int_{\mathcalG_\infty} \Psi_\infty(\vec k,\vec x)e^{-it\lambda_\infty(\vec k)}\hatt{\Psi}_0(\vec k)\,d\vec k ,
\eeq
then this function solves the initial value problem \eqref{IVP}, where
\beq\label{Def:Psi-0}\Psi _0 (\vec x) =\frac{1}{2\pi}\int_{\mathcalG_\infty} \Psi_\infty(\vec k,\vec x)\hatt{\Psi}_0(\vec k)\,d\vec k  \eeq
and $\Psi _0 (\vec x) \in S_\infty \mathcal S=E_\infty\mathcal C_0^\infty$.  Obviously,  $S_\infty \mathcal S$ is dense in  $E_\infty L^2(\R^2)$.

The first step  of the proof is replacing $\mathcalG_\infty$ by a small neighborhood $\tilde \mathcalG_\infty $ and to estimate the resulting errors in the integrals. This is an important step, since  $\mathcalG_\infty$ is a closed Cantor-type set, while $\tilde \mathcalG_\infty $ is an open set.
The second step is integrating by parts in
an integral over $\tilde \mathcalG_\infty $ with the purpose of obtaining \eqref{ball-1}, the fact that $\tilde \mathcalG_\infty $ is open being used for handling boundary terms.
All further considerations
%%%%are mostly based on Lemma \ref{eigenvalue} and its analogue for Fourier coefficients of $\Psi_n$ (see Remark \ref{raz}) and thus
are essentially identical for the limit-periodic and quasi-periodic cases. The notations are mostly  identical, in situations where they are different we consider the limit-periodic case.

To get the lower bound \eqref{ball-1}, we first note that
$$ \|X\Psi\|^2_{L^2(\R^2)} \geq \|X\Psi\|^2_{L^2(B_R)} \geq \frac{1}{2}\|Xw\|^2_{L^2(B_R)}-\|X(\Psi-w)\|^2_{L^2(B_R)},$$
where $B_R$ is the open disc with radius $R$ centered at the origin, $R=c_0T$, $c_0$ to be chosen later, and $w(\vec x,t)$ is an approximation of $\Psi $ when $\mathcalG_\infty$ is replaced by its small neighborhood $\tilde \mathcalG_\infty $. Namely,
\beq\label{Def:w}
w(\vec x,t):=\frac{1}{2\pi}\int_{\wti{\mathcalG}_\infty}\Psi_\infty(\vec k,\vec x)e^{-it\lambda_\infty(\vec k)}\hatt{\Psi}_0(\vec k)\eta_\delta(\vec k)\,d\vec k,
\eeq
$\eta_\delta$ being a smooth cut-off function with support in a $\delta$-neighborhood $\wti{\mathcalG}_\infty$ of $\mathcalG_\infty$ and $\eta_\delta=1$ on $\mathcalG_\infty$.
The parameter $\delta$ $(0<\delta<1)$ will be chosen later to be sufficiently small and depend only on $\hatt\Psi_0$.  We take $\eta_\delta$  to be a convolution of a function
$\omega (\vec k/2\delta )$ with the characteristic function of the $\delta/2$-neighborhood of $\mathcalG_\infty$, where $\omega (\vec k )$ is a nonnegative ${\mathcal C_0^\infty}(\R^2)$-function with a support in the unit ball centered at zero and integral one. Then, $\eta _\delta \in {\mathcal C_0^\infty}(\R^2)$,
\begin{equation}
0\leq \eta _\delta \leq 1,\  \eta _\delta(\vec k)=1\; \mbox{when } \vec k\in \mathcalG_\infty, \ \eta _\delta(\vec k)=0\;  \mbox{when } \vec k  \not \in \wti{\mathcalG}_\infty,\
\|D^{m}\eta _\delta \|_{L^{\infty }}<C_{m}\delta ^{-|m |}. \label{eta-delta}
\end{equation}

To prove \eqref{ball-1}, we will show that there exist a positive constant $c_1$ and  constants $c_2$ and $c_3$  such that
    \beq\label{Ineq:main}
    \frac{2}{T}\int _0^\infty e^{-2t/T} \big\|Xw(\cdot, t)\big \|^2_{L^2(B_R)} dt \geq 6c_1 T^2 - c_2 T - c_3,
    \eeq
as long as $c_0$ in the definition of $R$ exceeds a certain value depending only on $\hatt \Psi _0$. In formula  \eqref{Ineq:main}, the constant $c_1=c_1(\hatt \Psi_0)$ depends on $ \hatt \Psi_0$, but not $\delta $ or $c_0$, while the constants $c_2=c_2(\hatt \Psi_0,\delta )$ and $c_3=c_3(\hatt \Psi_0,\delta )$ depend on $ \hatt \Psi_0$ and $\delta $, but not $c_0$.

We also prove that
    \beq\label{Ineq:remainder}
    \frac{2}{T}\int _0^\infty e^{-2t/T} \big\|X(\Psi-w)(\cdot, t)\big \|^2_{L^2(B_R)} dt \leq \gamma(\delta, \hatt \Psi_0)c_0^2 T^2,
%%%%\text{ as } n \to \infty,
    \eeq
  $\gamma(\delta, \hatt \Psi_0)=o(1)$  as  $\delta\to 0$ uniformly in $c_0$.
%%%%\marginpar{we can even show $l.h.s(3.6)<c_{3,\alpha }T^{\alpha }$, $0<\alpha <1$}

%%%%%where $o(1)$ does not need to be uniform with respect to $T$.
%%%%\marginpar{here $n\to \infty$, not $T$!}
%Now, we consider $n$ satisfying \eqref{June12a-13} and \eqref{June12b-13}.

\begin{proof}[Proof of \eqref{Ineq:remainder}]
Since $\eta_\delta=1$ on $\mathcalG_\infty$,
 \begin{equation*}
 \Psi(\vec x,t)-w(\vec x,t)
 =-\frac{1}{2\pi}\int_{\wti{\mathcalG}_\infty\setminus \mathcalG_\infty}\Psi_\infty(\vec k,\vec x)e^{-it\lambda_\infty(\vec k)}\hatt{\Psi}_0(\vec k)\eta_\delta(\vec k)\,d\vec k=:f(\vec x,t).
 \end{equation*}
Since $\|X\| \leq R$, it suffices to show that
\beq \label{Ineq:f_{n,j}}
\|f(\cdot,t)\|_{L^2(\R^2)}^2  \leq \gamma (\delta, \hatt \Psi_0).
\eeq
Note that $f=S_\infty(\wti\mathcalG_\infty)g_1$, where  $g_1=e^{-it\lambda_\infty(\vec k)}\hatt{\Psi}_0(\vec k)\eta_\delta(\vec k)\chi(\wti\mathcalG_\infty \setminus \mathcalG_\infty)$ and
$S_{\infty}$ is defined by \eqref{Sinf}. Now, the estimate \eqref{S_infty} and Lebesgue's Dominated Convergence Theorem complete the proof, where Lemma~\ref{Ttransform} is used to show that
$\hatt{\Psi}_0(\vec k)$ decays sufficiently fast at infinity.

\end{proof}

%%%%%%%%%%%%%%%%%%%%%%%%%%%%%%%%%%%%%%%%%%%%%%%%%%%%%%%%%%%%%%%%
\begin{proof}[Proof of \eqref{Ineq:main}]
%%%%%Assume $\hatt{\Psi}_0(\vec k)\in \mathcalC^2(\mathcalG)_n)$ and its support is bounded.
Let
 \begin{equation}\label{v}
 v( \vec x,t):=\frac{1}{2\pi} \int_{\wti{\mathcalG}_\infty}\Psi_\infty(\vec k,\vec x)e^{-it\lambda_\infty(\vec k)}
 \nabla \Big(\hatt{\Psi}_0(\vec k)\eta_\delta(\vec k)\Big)\,d\vec k.
 \end{equation}
Then, using integration by parts and then \eqref{aplane}, we get
    \begin{align*}
    v( \vec x,t)
    & = -\frac{1}{2\pi}\int_{\wti{\mathcalG}_\infty}
    \Big[\nabla_{\vec k}\Psi_\infty(\vec k,\vec x)-it\Psi_\infty(\vec k,\vec x)
    \nabla \lambda_\infty(\vec k)\Big]e^{-it\lambda_\infty(\vec k)}
    \hatt{\Psi}_0(\vec k)\eta_\delta(\vec k)\,d\vec k \\
    & = -\frac{i}{2\pi}\int_{\wti{\mathcalG}_\infty} \Big[\vec x-t\nabla \lambda_\infty(\vec k)\Big]\Psi_\infty(\vec k,\vec x)
    e^{-it\lambda_\infty(\vec k)}\hatt{\Psi}_0(\vec k)\eta_\delta(\vec k)\,d\vec k \\
    & \hskip 0.5cm -\frac{1}{2\pi}\int_{\wti{\mathcalG}_\infty} e^{i\la \vec k, \vec x\ra}\nabla_{\vec k} u_\infty(\vec k,\vec x) e^{-it\lambda_\infty(\vec k)}\hatt{\Psi}_0(\vec k)\eta_\delta(\vec k)\,d\vec k,
    \end{align*}
where the boundary term is vanishing due to $\eta_\delta$ and the fast decay of $\hatt{\Psi}_0$. In short,
 $v=-iXw+it\phi-\phi_s,$ where
%By a simple calculation, we obtain
    %$$\|v\|^2_{L^2(B_R)} \geq \|Xw\|^2_{L^2(B_R)}(1-\veps_1)(1-N)
    %+t^2\|\phi\|^2_{L^2(B_R)}(1-\veps_1)\left(1-\f{1}{N}\right)+\|%\phi_s\^2_{L^2(B_R)}\left(1-\f{1{\veps_1}\right),$$
%for any $0 < \veps_1 <1$ and $N>1$, where
 \begin{align} \label{eq:phi}
    \phi(\vec x,t)&:=\frac{1}{2\pi} \int_{\wti{\mathcalG}_\infty} \nabla \lambda_\infty(\vec k)\Psi_\infty(\vec k,\vec x)
    e^{-it\lambda_\infty(\vec k)}\hatt{\Psi}_0(\vec k)\eta_\delta(\vec k)\,d\vec k \quad \text{ } \\
    \phi_s(\vec x,t)&:= \frac{1}{2\pi}\int_{\wti{\mathcalG}_\infty} e^{i\la \vec k, \vec x\ra} \nabla_{\vec k} u_\infty(\vec k,\vec x) e^{-it\lambda_\infty(\vec k)}\hatt{\Psi}_0(\vec k)\eta_\delta(\vec k)\,d\vec k, \notag
 \end{align}
and, therefore, $\|Xw\|^2_{L^2(B_R)}>\frac {t^2}{3}\|\phi \|^2_{L^2(B_R)}-\|v\|^2_{L^2(B_R)}-\|\phi _s\|^2_{L^2(B_R)}$. Integrating the last inequality with respect to $t$, we obtain:
    \begin{align}\label{ineq:lower-bound}
    & \f{2}{T} \int _0^\infty e^{-2t/T}  \|Xw\|^2_{L^2(B_R)}\,dt \notag\\
    & \ge \f{1}{3} \cdot \f{2}{T}\int _0^\infty t^2e^{-2t/T} \|\phi\|^2_{L^2(B_R)}\,dt -\f{2}{T}\int _0^\infty e^{-2t/T} \|v\|^2_{L^2(B_R)} \,dt \notag\\
    & \hskip 0.5cm-\f{2}{T}\int _0^\infty e^{-2t/T} \|\phi_s\|^2 \,dt =:\f{1}{3}I_1- I_2-I_3.
    \end{align}
Now we show that:
    \begin{align}
    I_1 & \ge 18(c_1 T^2 -c_2T),    %%%%%\text{ for any } 0<\veps_3<1,
    \label{ineq:I_1}\\
    I_2 & \le c\delta^{-2} \|\hatt{\Psi}_0\|_{W_2^1(\R^2)}^2, \label{ineq:I_2}\\
    I_3 & \le c(V)\|\hatt{\Psi}_0\|_{L^2(\R^2)}^2. \label{ineq:I_3}
    \end{align}
Let us prove  \eqref{ineq:I_2} first. From \eqref{v}, we see that $v=S_\infty(\wti \mathcalG_\infty)g_2$, where $g_2=e^{-it\lambda_\infty(\vec k)}\nabla \bigl(\hatt{\Psi}_0(\vec k)\eta_\delta(\vec k)\bigr)$, and, therefore, by \eqref{S_infty}, we get
 \begin{align*}
 \|v\|_{L^2(\R^2)} & \le 2\|\nabla \Big(\hatt{\Psi}_0\eta_\delta \Big)\|_{L^2(\R^2)} \\
 & \le 2\Big(\|\nabla \hatt{\Psi}_0\|_{L^2(\R^2)} + \|\hatt{\Psi}_0\|_{L^2(\R^2)}\|\nabla \eta_\delta\|_{L^\infty(\R^2)} \Big)\\
 & \le c\delta^{-1}\|\hatt{\Psi}_0\|_{W_2^1(\R^2)}
 \end{align*}
Now \eqref{ineq:I_2} is obvious.

Estimate \eqref{ineq:I_3} can be obtained in the same way as \eqref{June5} or \eqref{2.36}  with $\nabla_{\vec k}u_\infty$ instead of $u_\infty$ (see Appendix 3 for details).

Finally, we show the estimate \eqref{ineq:I_1}. Substituting \eqref{aplane} into \eqref{eq:phi} yields
 \begin{align*}
  \phi(\vec x,t)
 &= \frac{1}{2\pi}\int_{\wti{\mathcalG}_\infty} \nabla \lambda_\infty(\vec k) e^{i\la \vec k,\vec x\ra}
    e^{-it\lambda_\infty(\vec k)}\hatt{\Psi}_0(\vec k)\eta_\delta(\vec k)\,d\vec k\\
 &+\frac{1}{2\pi}\int_{\wti{\mathcalG}_\infty} \nabla \lambda_\infty(\vec k) e^{i\la \vec k,\vec x\ra}  u_\infty(\vec k,\vec x)
    e^{-it\lambda_\infty(\vec k)}\hatt{\Psi}_0(\vec k)\eta_\delta(\vec k)\,d\vec k\\
 &=: \wti\phi(\vec x,t)+\wti\phi_s(\vec x,t).
 \end{align*}
We use
 \begin{equation*}
\|\phi\|^2_{L^2(B_R)} \ge \frac12\|\wti\phi\|^2_{L^2(B_R)}-\|\wti\phi_s\|_{L^2(\R^2)}^2=\frac12\|\wti\phi\|^2_{L^2(\R^2)}-\frac12\|\wti\phi\|^2_{L^2(\R^2\setminus B_R)}-\|\wti\phi_s\|_{L^2(\R^2)}^2.
\end{equation*}
Thus,
\begin{equation}\label{RR}
\begin{split}
&\frac{2}{T}\int_0^\infty t^2 e^{-2t/T}\|\phi\|^2_{L^2(B_R)}\,dt=\frac{2}{T}\int_0^\infty t^2 e^{-2t/T}\left(\frac12\|\wti\phi\|^2_{L^2(\R^2)}-\|\wti\phi_s\|_{L^2(\R^2)}^2\right)\,dt\cr &-\frac{1}{T}\int_0^\infty t^2 e^{-2t/T}\|\wti\phi\|^2_{L^2(\R^2\setminus B_R)}\,dt=:R_1-R_2.
\end{split}
\end{equation}
To get a lower bound for $R_1$, we notice that
\begin{equation}\label{R1new}
\frac12 \|\wti \phi \|_{L^2(\R^2)}^2 - \|\wti \phi_s \|_{L^2(\R^2)}^2=\frac12 \|S_0({\widetilde{{\mathcal G}}_\infty})g_3\|^2_{L^2(\R^2)}-\|(S_\infty({\widetilde{{\mathcal G}}_\infty})-S_0({\widetilde{{\mathcal G}}_\infty}))g_3\|^2_{L^2(\R^2)}
 %4\pi^2 \int_{\wti \mathcalG_\infty} |\nabla \lambda_\infty(\vec k)|^2 |\hatt \Psi_0(\vec k)|^2 \eta_\delta(\vec %k)^2 \,d \vec k,
\end{equation}
where $g_3(\vec k):=\nabla \lambda_\infty(\vec k) \hatt{\Psi}_0(\vec k)\eta_\delta(\vec k)$. Now, using \eqref{2.36} with $c(V)\lambda_{**}^{-\gamma _6}\leq 1/4$ and noticing that $S_0$ is just the Fourier transform, we get
\begin{equation}\label{R1new1}
\begin{split}
&\frac12 \|\wti \phi \|_{L^2(\R^2)}^2 - \|\wti \phi_s \|_{L^2(\R^2)}^2\geq (\frac12-\frac14)\|g_3\|_{L^2(\wti \mathcalG_\infty)}^2\cr & = \frac14\int_{\wti \mathcalG_\infty} |\nabla \lambda_\infty(\vec k)|^2 |\hatt \Psi_0(\vec k)|^2 \eta_\delta(\vec k)^2 \,d \vec k\geq \frac14\int_{\mathcalG_\infty} |\vec k|^2 |\hatt \Psi_0(\vec k)|^2  \,d \vec k.
\end{split}
\end{equation}
Here we also used that on $\mathcalG_\infty$ we have $\eta_\delta=1$ and $|\nabla \lambda_\infty| \geq |\vec k|$.

The bound \eqref{R1new1} immediately implies the main estimate of the paper:
\begin{equation}\label{R1}
R_1\geq \frac18 T^2\int_{\mathcalG_\infty} |\vec k|^2 |\hatt \Psi_0(\vec k)|^2  \,d \vec k=:20c_1T^2,
\end{equation}
\begin{equation}\label{c_1}
c_1=c_1(\Psi_0):=\frac{1}{160}\int_{\mathcalG_\infty} |\vec k|^2 |\hatt \Psi_0(\vec k)|^2  \,d \vec k.
\end{equation}
For $R_2$, let us introduce a new variable $\vec z:=\vec x/t$ and consider
 \begin{equation}\label{Nov-28-1}
 \wti\phi(\vec zt,t)=
 \frac{1}{2\pi}\int_{\wti{\mathcalG}_\infty} e^{it\left(\la \vec k,\vec z\ra
    -\lambda_\infty(\vec k)\right)}g_3(\vec k)\, d\vec k.
 \end{equation}
%%where $g(\vec k):=\nabla \lambda_\infty(\vec k)\hatt{\Psi}_0(\vec k)\eta_\delta(\vec k)$.
We use the method of stationary phase and integration by parts.  Considering \eqref{May31} and Lemma~\ref{lambda}, we conclude that the equation for a stationary point
%%%%For each $\vec z$, solve
 $$
 \vec z- \nabla \lambda_\infty (\vec k) =0
 $$
 has a unique solution $\vec k_0(z):=\vec k_0$  and
%%then by Lemma \ref{lambda},
%%there is a unique stationary point $\vec k$ denoted by $\vec k$ which can be written as
 $$
 \vec k_0 = \frac{1}{2} \vec z + O(|\vec z|^{-\gamma_5}),\ \ \gamma _5>0.
 %%%%% \quad \gamma_5=\gamma_2-\gamma_0=1-46s_1-31\delta .
 $$
Let $\eta$ be a smooth cut-off function satisfying
 $$
 \eta(\vec k)=
  \begin{cases}
   0, & \left| \vec k -\vec k_0 \right| \le 1 \\
   1, & \left| \vec k -\vec k_0 \right| \ge 2
  \end{cases}
 .
 $$
 Then,
 \begin{align}\label{Nov-28-2}
 \wti\phi(\vec zt,t)&=
\frac{1}{2\pi} \int_{\wti{\mathcalG}_\infty \cap \{\vec k \, :\, \left| \vec k - \vec k_0 \right| < 2 \} } e^{it\left(\la \vec k,\vec z\ra
    -\lambda_\infty(\vec k)\right)}g_3(\vec k)\bigl(1-\eta(\vec k)\bigr)\, d\vec k \\
    &+  \frac{1}{2\pi} \int_{\wti{\mathcalG}_\infty \cap \{\vec k:\, \left| \vec k - \vec k_0 \right| > 1 \} } e^{it\left(\la \vec k,\vec z\ra
    -\lambda_\infty(\vec k)\right)}g_3(\vec k) \eta(\vec k)\, d\vec k\\
    &=: \wti\phi_1(\vec zt,t)+\wti\phi_2(\vec zt,t).
 \end{align}
To estimate $\wti\phi_1(\vec zt,t)$, we first note that $g_3(1-\eta) \in {\mathcal C}_0^4(\R^2)$ and   $ \la \vec k,\vec z\ra
    -\lambda_\infty(\vec k) \in {\mathcal C^7}(\R^2)$, the estimate \eqref{eigenvalue} holding for $|m|\leq 7$ with $-\gamma _2+7\gamma _0<0$.
%%%Moreover, by \eqref{eigenvalue},
%% $$
%% \|g_3(1-\eta)\|_{{\mathcal C^4}(\R^2)} < c \delta^{-4} \| \vec k \hatt \Psi_0\|_{{\mathcal C^4}(\R^2)}.
%%%% c \delta^{-4} k_{**}^{-\gamma_2+5\gamma_0} \| \hatt \Psi_0\|_{{\mathcal C^4}(\R^2)}.
%% $$
Therefore, applying Theorem~7.7.5 in \cite{H1} yields:
\begin{equation}\label{tilde_phi_1}
 \wti\phi_1(\vec zt,t)=\frac{1}{2i} e^{it\left(\la \vec k_0,\vec z\ra -\lambda_\infty(\vec k_0)\right)}
 \left(1+O(|\vec z|^{-\gamma_5})\right)g_3(\vec k_0)t^{-1} + \epsilon(g_3)t^{-2}
 \end{equation}
 for $|z|^2>\lambda_*$ and $0$ otherwise. Here
 $$
 |\epsilon(g_3)| \leq c \sum_{{|m| \leq 4}}\sup_{\left| \vec k -\vec k_0 \right| <2 } |D^{m}g_3(\vec k)| \leq c\left\||\vec k|^3\hatt \Psi_0(\vec k)\right\|_{\mathcal C^4(\R^2)} \delta ^{-4}|\vec z|^{-2},
 $$
 Next, we consider  $\wti\phi_2 (\vec zt,t)$. There is no stationary point. Integrating by parts twice, we obtain
 \begin{equation}\label{tilde_phi_2}
 |\wti\phi_2(\vec zt,t)| \leq C(\hatt\Psi _0) (\delta t)^{-2}(1+|\vec z|)^{-2},
 \end{equation}
 where $C(\hatt\Psi _0)$ is a combination of integrals of the type $\int |\vec k|^j|D^{m }\hatt\Psi _0 (\vec k)|d\vec k $, $0\leq j \leq 3$, $0\leq |m |\leq 2$.

Now, we consider $\|\wti\phi(\vec x,t)\|_{L^2(\R^2\setminus B_R)}^2$. Using the estimates \eqref{tilde_phi_1} and
 \eqref{tilde_phi_2}, we obtain
 $$
 \|\wti\phi(\vec x,t)\|_{L^2(\R^2\setminus B_R)}^2
 = t^2 \|\wti\phi(\vec zt,t)\|_{L^2(\R^2\setminus B_{R/t})}^2
 \leq \int _{\R^2\setminus B_{c_0T/t}}|g_3(\vec k_0 (\vec z))|^2\, d\vec z +O(t^{-1})$$
as $t\to \infty $, the constant in $O(t^{-1})$ depending on $\delta$ and $\hatt\Psi _0$.
Next, substituting the above estimate   into the formula for $R_2$ (see \eqref{RR}) and changing the variables $s=t/T$, we obtain:
 \begin{align}\label{R_2}
R_2 &
%%%%%T^2 \int _0^\infty s^2e^{-2s} \|\wti\phi(\cdot, Ts)\|_{L^2(\R^2\setminus B_R)}^2 \,ds \notag\\ &
  \leq T^2 \int _0^\infty s^2e^{-2s}\int _{\R^2\setminus B_{c_0/s}}|g_3(\vec k_0 (\vec z))|^2\, d\vec z  \,ds +O(T).
 \end{align}
By  Lebesgue's Dominated Convergence Theorem, the integral on \eqref{R_2} goes to zero when $c_0 \to \infty$ uniformly in $\delta $.
We choose $c_0$ large enough to ensure that
 $$
R_2\leq {c_1} T^2 +c T,
 $$
the constant $c_1$ being defined by \eqref{c_1}. Notice that the choice of $c_0$ depends  on $\hatt\Psi_0$, but not $\delta $.
Considering the last estimate together with \eqref{R1}, we obtain \eqref{ineq:I_1}.

\end{proof}

%%%%%%%%%%%%%%%%%%%%%%%%%%%%%%%%%%%%%%%%%%%%%%%%%%%%%%%%%%%%%%%%
\begin{proof}[Proof of \eqref{ball-1}]
After $c_0$ is fixed as above we choose a sufficiently small $\delta=\delta (c_0,\hatt \Psi _0)$ so that the constant $\gamma c_0^2$ from \eqref{Ineq:remainder} is smaller than $c_1$. Thus, we obtain:
 \begin{equation}\label{ball-2}
 \frac{2}{T}\int _0^\infty e^{-2t/T} \big\|X \Psi (\cdot,t)\big \|^2_{L^2(\R^2)} dt >2c_1(\Psi _0)T^2-c_2(\Psi _0)T-c_3(\Psi _0),  \ \ c_1>0.
 \end{equation}
Taking $T$ sufficiently large, we obtain \eqref{ball-1} for any non-zero $\Psi _0\in E_{\infty }{\mathcal C_0^\infty}$.

%Equivalently, ${\mathcal A}_{\infty }:=E_{\infty }\wti {\mathcal S}$, where $\wti {\mathcal S}$ is the set of functions $f(\vec x)$:
%$\|f\|_{W_2^2(\R^2)}+\||\vec x|^6f\|_{L^2(\R^2)}<\infty $.
\end{proof}
%%%%%%%%%%%%%%%%%%%%%%%%%%%%%%%%%%%%%%%%%%%%%%%%%%%%%%%%%%%%%%%%%

%Thus, we proved the theorem for $\hatt \Psi _0$ with a supposrt in $K(k_0)$. Assume $\hatt \Psi _0$ is in $\mathcal C^2$ and has a support in an arbitrary bounded set $K$, $K\cap {\mathcal G}_{\infty }\neq 0$. Using  a partition of unity we represent $\hatt \Psi _0$ as a finite sum $\hatt \Psi _0=\sum _i \hatt \Psi _0^{(i)}$, where $\Psi _0^{(i)}$ has a support
%in a set $K(k_0^{(i)}$, $K\subset \cup _i K(k_0^{(i)})$. Summarizing \eqref{Ineq:remainder} for each $\hatt \Psi _0^{(i)}$ we obtain the estimate \eqref{Ineq:remainder}  for $\hatt \Psi _0$. Summarizing \eqref{ineq:I_2}, \eqref{ineq:I_3} for each $\hatt \Psi _0^{(i)}$ we obtain similar estimates for $\hatt \Psi _0$. it remains to check   \eqref{ineq:I_1}.
%Again, summarizing \eqref{ineq:wtiphi_s} for each $\hatt \Psi _0^{(i)}$ we obtain the estimate \eqref{ineq:wtiphi_s}   for $\hatt \Psi _0$. In the proof of \label{ineq:wtiphi}, we really
%did not use the assumption that the support of $\hatt \Psi _0$ belongs to $K(k_0)$. It  works for $\hatt \Psi _0$ with any bounded support. Thus,  \eqref{ineq:I_1} works for
%$\hatt \Psi _0$. Theorem 1 follows.

\section{Proof of Theorem \ref{Thm1}} \label{Thm1proof}

Now we prove Theorem~\ref{Thm1}.
Let $\Psi_0 \in {\mathcal C_0^\infty}(\R^2)$ then by Lemma~\ref{Ttransform}, $ \hatt\Psi_0 \in {\mathcal C^L}$ decays faster than any polynomials of degree at most $L$,
where
 \begin{equation}
 \hatt\Psi_0 (\vec k) = (T_\infty \Psi_0)(\vec k)= \frac{1}{2\pi}\int_{\R^2} \overline{\Psi_\infty (\vec k , \vec x)} \Psi_0(\vec x)\,d\vec x .
 \label{hat}\end{equation}
We denote
 $$
 \Psi_{0, \text{ac}} := E_\infty(\mathcalG_\infty)\Psi_0 = \frac{1}{2\pi}\int_{\mathcalG_\infty} \Psi_\infty (\vec k , \vec x) \hatt\Psi_0(\vec k)\,d\vec k
 $$
and
 $$
 \Psi_{0, \text{s}} := \Psi_0-\Psi_{0, \text{ac}} .
 $$
We notice that $\Psi_{0, \text{s}}\perp E_\infty L^2(\R^2)$ and $\| \Psi_{0, \text{s}}\|_{L^2(\R^2)}\leq \| \Psi_{0}\|_{L^2(\R^2)}$.
Assume that $\Psi_{0, \text{ac}}$ is not identically zero. We put
 $$
 \Psi(\vec x ,t)= \Psi_{\text{ac}}(\vec x ,t) + \Psi_{\text{s}}(\vec x ,t) := e^{-itH}\Psi_{0, \text{ac}} +  e^{-itH}\Psi_{0, \text{s}}.
 $$
As in the proof of Theorem \ref{Thm1}, we use
 $$
 \| X \Psi\|_{L^2(\R^2)} \geq  \| X \Psi\|_{L^2(B_R)}
 $$
and approximate $\Psi_{\text{ac}}$ by $w$ defined as in \eqref{Def:w}.
Next, we rewrite
 $$
 \|X(\Psi_{\text{s}}+w)\|_{L^2(B_R)}^2=\|X\Psi_{\text{s}}\|_{L^2(B_R)}^2+\|Xw\|_{L^2(B_R)}^2+2\Re(X\Psi_{\text{s}},Xw)_{L^2(B_R)}.
 $$
%%5It suffices to show that
 Let us note that $(X\Psi_{\text{s}},Xw)_{L^2(B_R)}=(\Psi_{\text{s}},X^2w)_{L^2(B_R)}$
 %%%%We will show that the contribution of this term is relatively small. Indeed,  let
 and consider its integral over $t$:
 \begin{equation}\label{cross}
\hat I:= \frac{2}{T}\int_0^\infty  e^{-2t/T}|(\Psi_{\text{s}},X^2w)_{L^2(B_R)}|\, dt .
 %%%%%= \tilde\gamma(\delta,c_0,T) T^2,
 \end{equation}
%%%%%where $\tilde\gamma(\delta,c_0,T)$ can be made arbitrary small by choosing appropriate small $\delta$ and large $c_0$ for any sufficiently large $T$.
%As before, first, we choose large enough $c_0=c_0(\Psi_0,k_0)$ then small %$\delta=\delta(c_0,\Psi_0,k_0)$ and finally $T$ will be larger than some $T_0=T_0(\delta, c_0,\Psi_0,k_0)$.
%%%We consider $X^2w$ as follows.
Considering \eqref{Ineq:main},
%%%%and the obvious inequality $\|\Psi_{\text{s}}(\cdot ,t)\|_{L^2(\R^2)}|\leq \| \Psi_{0}\|_{L^2(\R^2)}$,
we see that it is enough to show that $\hat I$ is small compared with the r.h.s.\ of  \eqref{Ineq:main}. We achieve this by proving that $X^2 w$ is orthogonal to $\Psi_{\text{s}}$ up to minor terms. Indeed, by \eqref{Def:w} and \eqref{aplane},
 \begin{align*}
 (X^2w)(\vec x, t)
 &=\frac{1}{2\pi} \int_{\tilde\mathcalG_\infty} |\vec x|^2\Psi_\infty(\vec k,\vec x)e^{-it\lambda_\infty(\vec k)}\hatt{\Psi}_{0}(\vec k)\eta_\delta(\vec k) \,d\vec k\\
 &=-\frac{1}{2\pi}\int_{\tilde\mathcalG_\infty} (\Delta_{\vec k} e^{i \la \vec k, \vec x \ra})(1+ u_\infty(\vec k, \vec x))
   e^{-it\lambda_\infty(\vec k)}\hatt{\Psi}_{0}(\vec k)\eta_\delta(\vec k)\,d\vec k .
 \end{align*}
Using $g_4(\vec k, \vec x):=(1+ u_\infty(\vec k, \vec x))\hatt{\Psi}_{0}(\vec k)\eta_\delta(\vec k)$  and applying integration by parts as above, we obtain
 \begin{align*}
 (X^2w)(\vec x, t)
 &=t^2 \frac{1}{2\pi}\int_{\tilde \mathcalG_\infty} e^{i \la \vec k, \vec x \ra} \left|\nabla \lambda_\infty (\vec k)\right|^2 e^{-it\lambda_\infty(\vec k)} g_4(\vec k, \vec x) \,d\vec k\\
 &\hskip .5cm - \frac{1}{2\pi}\int_{\tilde \mathcalG_\infty} e^{i \la \vec k, \vec x \ra} e^{-it\lambda_\infty(\vec k)} \Delta_{\vec k}g_4(\vec k, \vec x) \,d\vec k\\
 &\hskip .5cm +t\frac{i}{\pi} \int_{\tilde\mathcalG_\infty} e^{i \la \vec k, \vec x \ra} e^{-it\lambda_\infty(\vec k)}
  \left\la \nabla \lambda_\infty(\vec k),\nabla_{\vec k} g_4(\vec k, \vec x) \right\ra  \,d\vec k\\
 &\hskip .5cm +t \frac{i}{2\pi}\int_{\tilde\mathcalG_\infty} e^{i \la \vec k, \vec x \ra} \left(\Delta \lambda_\infty(\vec k)\right) e^{-it\lambda_\infty(\vec k)} g_4(\vec k, \vec x) \,d\vec k .
 \end{align*}
The last three integrals can be estimated as in the proof of \eqref{Ineq:main} (see \eqref{ineq:I_2},\eqref{ineq:I_3}), and the corresponding contribution to $\hat I$ is bounded by a linear function of $T$ for every fixed $\delta >0$. For the first integral, we have
 \begin{align*}
 &t^2\frac{1}{2\pi}\int_{\tilde\mathcalG_\infty} e^{i \la \vec k, \vec x \ra}  \left|\nabla \lambda_\infty (\vec k)\right|^2 e^{-it\lambda_\infty(\vec k)}g_4(\vec k, \vec x) \,d\vec k\\ &=t^2\frac{1}{2\pi}\int_{\mathcalG_\infty}\Psi_\infty(\vec k,\vec x)  \left|\nabla \lambda_\infty (\vec k)\right|^2 e^{-it\lambda_\infty(\vec k)}\hatt{\Psi}_{0}(\vec k) \,d\vec k\\
 &\hskip .5cm  +t^2\frac{1}{2\pi}\int_{\tilde\mathcalG_\infty\setminus\mathcalG_\infty} \Psi_\infty(\vec k,\vec x) \left|\nabla \lambda_\infty (\vec k)\right|^2
 e^{-it\lambda_\infty(\vec k)}\hatt{\Psi}_{0}(\vec k) \eta_\delta(\vec k) \,d\vec k\\
 &=:t^2(J_1+J_2).
 \end{align*}
Obviously, $\|J_2\|_{L^2(\R^2)}=o(1)$ as $\delta\to 0$ uniformly in $t$ (cf. \eqref{Ineq:f_{n,j}}) and its contribution to $\hat I$ is bounded by $\gamma T^2$, where $\gamma(\delta,\Psi_0)\to 0$ as $\delta\to 0$ . To estimate the contribution from $J_1$ we notice that $J_1=E_\infty(\mathcalG_\infty)J_1$ and, thus, we arrive at the main point of the proof:
 $$
 (\Psi_{s},J_1)_{L^2(B_R)}=-(\Psi_{s},J_1)_{L^2(\R^2\setminus B_R)}.
 $$
%We also notice that $\|\Psi_{\text{s}}\|_{L^2(\R^2)}\leq C(\Psi_0)$.
It remains to estimate
\begin{equation}\label{I1} \hat I_1=\frac{2}{T}\int_0^{\infty}t^2e^{-2t/T}|(\Psi_{s},J_1)_{L^2(\R^2\setminus B_R)}|dt. \end{equation}
It is easy to see that
 \begin{equation}\label{last}
  \begin{split}
  %%&\hat I_1=\frac{2}{T}\int_0^{\infty}t^2e^{-2t/T}|(\Psi_{s},J_1)_{L^2(\R^2\setminus B_R)}|dt \\
  &\hat I_1\leq \frac{2}{T}\int_0^{\infty}t^2e^{-2t/T}(\epsilon\|\Psi_{s}\|^2_{L^2(\R^2\setminus B_R)}+\frac{1}{4\epsilon}\|J_1\|^2_{L^2(\R^2\setminus B_R)})dt\\
  &\leq \epsilon C(\Psi_0)T^2+\frac{1}{\epsilon T}\int_0^{\infty}t^2e^{-2t/T}(\| J_1 + J_2 \|^2_{L^2(\R^2\setminus B_R)}+\|J_2\|^2_{L^2(\R^2\setminus B_R)})dt.
  \end{split}
\end{equation}
%Here
%$$
%\tilde J_1:=\int_{\tilde\mathcalG_\infty}\Psi_\infty(\vec k,\vec x) (\Delta_{\vec k}\lambda_\infty(\vec k))e^{-it\lambda_\infty(\vec k)}\hatt{\Psi}_0(\vec k)\eta\,d\vec k
%$$
%and
%$$
%J_{1,\delta}:=\int_{\tilde\mathcalG_\infty\setminus\mathcalG_\infty}\Psi_\infty(\vec k,\vec x) (\Delta_{\vec k}\lambda_\infty(\vec k))e^{-it\lambda_\infty(\vec k)}\hatt{\Psi}_0(\vec %k)\eta\,d\vec k.
%$$
The estimate for the integral with  $J_1+J_2$ is similar to the estimate for $R_2$ (see \eqref{R_2}), while the estimate for the integral with $J_2$ repeats the proof for \eqref{Ineq:f_{n,j}}. Thus, \eqref{last} is bounded by
 $$
 \epsilon C(\Psi_0)T^2+\frac{1}{2\epsilon}\left(T^2\hat \gamma(c_0,\Psi_0)+C(\Psi_0,\delta )T+T^2 \gamma (\delta,\Psi_0)\right),
 $$
where $\hat \gamma(c_0,\Psi_0)\to 0$ as $c_0 \to \infty$ and $\gamma (\delta,\Psi_0)\to 0$ as $\delta\to 0$.
Now, one chooses small $\epsilon$, then large $c_0$, small $\delta$ and large $T_0$ to prove \eqref{cross}.

\begin{remark} \label{smoothEFE}

(a) The above proofs show that Theorem~\ref{Thm1} remains true if we replace $C_0^{\infty}$ in (\ref{nontriv}) with
\[ {\mathcal S}_8 := \{f:|x|^sD^mf(x) \in L^2(\R^2), 0\leq s,|m| \leq 8\}\]
 i.e.\ for initial conditions which are sufficiently smooth and of sufficiently rapid power decay. This is a consequence of the fact that the assumption of Lemma~\ref{Ttransform} can be weakened accordingly, see the proof in Section~\ref{App2} below.

(b) Using the constructions in the above proofs, we can now also describe more explicitly how to choose initial conditions $\Psi_0$ for the solution of (\ref{IVP}) which give simultaneous ballistic upper and lower bounds. Essentially, one has to regularize elements in the range of $E_{\infty}$ in two different ways, one at the boundary of ${\mathcal G}_{\infty}$, using the cut-off function $\eta_{\delta}$ as in (\ref{eta-delta}) above, and one at high momentum $\vec{k}$. For the latter, let $\varphi \in {\mathcal S}_8$ on $\R^2$  and such that $\varphi$ does not vanish identically on ${\mathcal G}_{\infty}$.

Choose
\begin{equation}
\Psi_0(\vec x) :=\frac{1}{2\pi}\int_{\wti{\mathcalG}_\infty} \varphi(\vec k) \,\eta_\delta(\vec k) \,\Psi_\infty(\vec k,\vec x) \,d\vec k.
\end{equation}
As $\delta\to 0$ this converges to $F_0(\vec x) = \frac{1}{2\pi} \int_{{\mathcal G}_{\infty}} \varphi(\vec x) \Psi_{\infty}(\vec k, \vec x)\, d\vec k$ in the range of $E_{\infty}$ with $\|F_0\|^2 = \int_{{\mathcal G}_{\infty}} |\varphi|^2\, d{\vec k}/(4\pi^2)\not=0$. Thus, for $\delta>0$ sufficiently small, $E_{\infty} \Psi_0 \not=0$.

Furthermore, our methods, in particular those provided in the Appendices in Section~\ref{appendices}, show that the choice of $\varphi \in {\mathcal S}_8$ gives $\Psi_0 \in {\mathcal S}_8$ . Thus the initial condition $\Psi_0$ leads to a ballistic lower bound on transport. At the same time the condition of \cite{RaSi} for the ballistic upper bound (\ref{genballistic}) is satisfied.

\end{remark}

\section{Appendices} \label{appendices}

Here we provide detailed proofs of some of the facts which were used in Sections~\ref{SpecPropH} and \ref{Prop2proof} above.

\begin{remark}\label{raz}
Using the a priori estimates (see \cite{KL3}) for the solutions $\Psi_n$ from \eqref{na} and their Fourier coefficients defined by
  \begin{align} \label{raz*}
  \Psi_n(\vec k, \vec x) &= e^{i \la \vec k,\vec x \ra}\left(1+u_n(\vec k, \vec x)\right), \\
  u_n(\vec k, \vec x) &= \sum_{r \in \Z^2} C_r^{(n)}(\vec k) e^{i \la \vec p_r^{(n)},\vec x \ra}, \label{raz**}
  \end{align}
 $\vec p_r^{(n)}$ being vectors of the dual lattice corresponding to $W_n$, and repeating the arguments which led to Lemma~\ref{lambda}, one can obtain that the extended coefficients are sufficiently smooth and satisfy estimates of the type  \eqref{eigenvalue-1}, \eqref{derivative-eigenvalue-n*}, \eqref{eigenvalue}. We omit the details.
\end{remark}
\subsection{Appendix 1}

\begin{lem} \label{Lem2}
The sequence of operators $S_n(\mathcal{G}_{\infty , \lambda
    })$ given by \eqref{SnGDef} has a  limit $S_{\infty }(\mathcal{G}_{\infty , \lambda
    })$ in the class of bounded operators.
    %%%%The operator $S_{\infty }(\mathcal{G}_{\infty , \lambda
  %%%%%  })$ satisfies $\|S_{\infty }\|\leq 1$ and can be described by the
%%%%%    formula
 %%%%   \begin{equation}
 %%%%%   (S_{\infty }f) (x)= \frac{1}{2\pi }\int _{\mathcal{G}_{\infty , \lambda
%%%%    }}f (\vec k)\Psi _{\infty }(\vec
  %%%%  \varkappa ,x) d\vec k  \label{ev1}
 %%%   \end{equation}
%%%for any $f \in L^{\infty }\left( \mathcal{G}_{\infty , \lambda
%%}\right)$.
The convergence of $S_n(\mathcal{G}_{\infty , \lambda
    })$ to $S_{\infty }(\mathcal{G}_{\infty , \lambda
    })$  is uniform in $\lambda $ and estimate \eqref{June5} holds.
    %%%for every $f \in L^{2}\left(
%%%    \mathcal{G}_{\infty }\right)$.
\end{lem}
\begin{proof}
It suffices to prove  that $S_n({{\mathcal G}}_{\infty , \lambda
    })f $ is a Cauchy sequence.
    %%%%in $L^2(\R^2)$ for every $f \in L^{\infty}\left(
 %%%%   \mathcal{G}_{\infty ,\lambda
  %%$$  }\right)$.
  Given $Q_n$ is the cell of periods of the operator
    $H^{(n)}$,
    %%%%%%see \eqref{a'},
     the function $\Psi _n(\vec k,x)$ is quasi-periodic in $Q_n$. It can be
    represented as a  combination of plane waves \eqref{raz*}, \eqref{raz**}.
   %%%%% \begin{equation}\Psi _n(\vec k,x)=
   %%%% \frac{1}{2\pi}\sum _{r\in \Z^2}c_r^{(n)}(\vec k)
  %%%%%  \exp i\langle \vec k+\vec p_r(0)/\tilde
%%%    N_{n},x\rangle,\label{++}
  %%%  \end{equation}
%%%%%where $c_r^{(n)}(\vec k)$ are Fourier coefficients and $\vec
%%p_r(0)=(\frac{2\pi r_1}{\beta _1}, \frac{2\pi r_2}{\beta _2})$.
The
Fourier transform of $\widehat \Psi_n$ is a combination of $\delta
$-functions:
    $$\widehat \Psi _n(\vec k,\vec \xi)=\sum _{r\in \Z^2}C_r^{(n)}(\vec k)
    \delta \bigl(\vec \xi +\vec k+\vec p_r(0)/\tilde N_{n}\bigr).$$
From this, we compute easily the Fourier
    transform of $S_nf $
    $$ (\widehat{S_nf})(\vec \xi)=\frac{1}{2\pi }\sum _{r\in \Z^2}C_r^{(n)}\bigl(-\vec \xi-
    \vec p_r(0)/\tilde N_{n}\bigr)f \bigl(-\vec \xi-
    \vec p_r(0)/\tilde N_{n}\bigr)\chi \bigl({\mathcal G}_{\infty ,\lambda},-\vec \xi-
    \vec p_r(0)/\tilde N_{n}\bigr),$$
where $\chi ({\mathcal G}_{\infty ,\lambda},\cdot ) $ is the
characteristic function on ${\mathcal G}_{\infty
    ,\lambda }$. Since ${\mathcal G}_{\infty ,\lambda }$ is bounded, the series contains only a
    finite number of non-zero terms for every $\vec \xi $.
    By Parseval's identity, triangle
    inequality and a parallel shift of the variable,
    $$
    \|S_nf\|_{L^2(\R^2)}=\|\widehat{S_nf}\|_{L^2(\R^2)}
    $$
    $$
     \leq \frac{1}{2\pi }\sum _{r\in \Z^2}\left \|C_r^{(n)}\bigl(-\vec \xi-
    \vec p_r(0)/\tilde N_{n}\bigr)f \bigl(-\vec \xi-
    \vec p_r(0)/\tilde N_{n}\bigr)\chi \bigl({\mathcal G}_{\infty , \lambda},-\vec \xi-
    \vec p_r(0)/\tilde N_{n}\bigr)\right\|_{L^2(\R^2)}$$ $$= \frac{1}{2\pi }\sum _{r\in \Z^2}\|C_r^{(n)}(\vec k )f (\vec k )\|_{L^2({\mathcal G}_{\infty ,
    \lambda})}.$$
  Assume first that the support of $f $ belongs to a ring
    ${R_{k,2k}}$ for some $k$ such that $k^{2}>\lambda
    _*(V)$.
    Then, the last inequality yields:
    \begin{equation}\|S_nf\|_{L^2(\R^2)}\leq \frac{1}{2\pi }\|f \|_{L^{2}(R_{k,2k})}\sum _{r\in \Z^2}\|C_r^{(n)}\|_{L^{\infty}(R_{k,2k})}.
    \label{appendix} \end{equation}
    %%%%% $$
   %%%%  \frac{1}{2\pi }\|f \|_{L^{\infty}(R_{k,2k})}\left(\sum _{r\in \Z^2}p_r^{4}(0)\|c_r^{(n)}\|^2_{L^2(R_{k,2k})}\right)^{1/2}\left(\sum_{r\in \Z^2}
   %%%  p_r^{-4}(0)\right)^{1/2},$$
 %%%    where we used Cauchy-Schwarz inequality.
  %%%%%%%% \begin{multline*}
 %%%%%%%  \|S_nf\|_{L^2(\R^2)}\leq \|f \|_{L^{\infty}({\mathcal G}_{\infty ,
  %%%%%%%  \lambda})}\sum _{r\in \Z^2}\|c_r^{(n)}\|_{L^2({\mathcal G}_{\infty ,
  %%%%%%%  \lambda})} \leq \\
  %%%%%%%   \|f \|_{L^{\infty}({\mathcal G}_{\infty ,
 %%%%%%%   \lambda})}\left(\sum _{r\in \Z^2}p_r^{2l}(0)\|c_r^{(n)}\|^2_{L^2({\mathcal G}_{\infty ,
 %%%%%%%   \lambda})}\right)^{1/2}\left(\sum_{r\in \Z^2} p_r^{-2l}(0)\right)^{1/2}.
   %%%% \end{multline*}
%%%%%%%Obviously, each term of the series can be bounded by
 %%%%%%%    $\|c_r^{(n)}\|_{L^{\infty}({\mathcal G}_{\infty ,\lambda})}\|f\|_{L^{\infty
%%%%%%%    }({\mathcal G}_{\infty ,\lambda})}|{\mathcal G}_{\infty,\lambda }|^{1/2}.$
By (\ref{raz**}),  Fourier coefficients $C_r^{(n)}(\vec k)$ can
be estimated as follows:
$$p_r^{4}(0)|C_r^{(n)}(\vec k)|
    \leq 2\pi \|\Psi _n(\vec k,\cdot )\exp \left(-i\langle \vec k,\cdot \rangle\right)\|_{W^{4}_2(Q_n)}
    |Q_n|^{-1/2}\tilde N_{n}^{4}$$ $$\leq 16\pi|\vec k|^{4}\|\Psi _n(\vec k,\cdot)\|_{W^{4}_2(Q_n)}
    |Q_n|^{-1/2}\tilde N_{n}^{4}
    .$$
  %%%  Integrating the last inequality over $R_{k,2k}$, we arrive at
%%%    $$\sum_{r \in \Z^2}p_r^{4}(0)\|c_r^{(n)}\|^2_{L^2(R_{k,2k})}
  %%  \leq ck^{4+2}|Q_n|^{-1}\tilde N_{n}^{4}\sup _{\vec k \in R_{k,2k}}
 %%   \|\Psi _n(\vec k,\cdot )\|^2_{W^{2l}_2(Q_n)}.
%%%%%%%    $$
%%%%%%%      $$\sum_{r \in \Z^2}p_r^{4}(0)|c_r^{(n)}(\vec k)|^2
%%%%%%%      \leq 4\pi ^2\|\Psi _n(\vec k,\cdot )\exp \left(-i\langle \vec k,\cdot \rangle\right)\|^2_{W^{2l}_2(Q_n)}
%%%%%%%      |Q_n|^{-1}\tilde N_{n}^{4}\leq $$ $$8\pi^2|\vec k|^{4}\|\Psi _n(\vec k,\cdot)\|^2_{W^{2l}_2(Q_n)}
%%%%%%%      |Q_n|^{-1}\tilde N_{n}^{4}
%%%%%%%      .$$
%%%%%%%      Integrating the last inequality over $R_{k,2k}$, we arrive at
%%%%%%%      $$\sum_{r \in \Z^2}p_r^{4}(0)\|c_r^{(n)}\|^2_{L^2(R_{k,2k})}
%%%%%%%      \leq ck^{4+2}|Q_n|^{-1}\tilde N_{n}^{4}\sup _{\vec k \in R_{k,2k}}
%%%%%%%      \|\Psi _n(\vec k,\cdot )\|^2_{W^{2l}_2(Q_n)}.
%%%%%%%      $$
    %%%%%%%%%$$\sup _{\vec k \in {\mathcal G}_{\infty ,\lambda}}2|\vec k|^{2l}\|\Psi _n(\vec k,\cdot)\|^2_{W^{2l}_2(Q_n)}
%%%%%%%%    |Q_n|^{-1}\tilde N_{n-1}^{2l}
%%%%%%%    .$$
    %%%%%%%%Considering also that the sum in $r$ contains no more than
   %%%%%%%% $|G_{\infty ,\lambda }||K_n|^{-1}$ nonzero terms,
Considering that $\sum_{r\neq 0} p_r^{-4}(0)<c$,
we obtain:
$$\sum_{r \in \Z^2}\|C_r^{(n)}\|_{L^{\infty}(R_{k,2k})}<c\sup _{\vec k \in R_{k,2k }}\left(|\vec k|^{4}\|\Psi _n(\vec k,\cdot)\|_{W^{4}_2(Q_n)}
    |Q_n|^{-1/2}\tilde N_{n}^{4}\right)
    .$$
Using \eqref{appendix}, we arrive at
   %%%%%%%%  estimate:
     $$ \|S_nf\|_{L^2(\R^2)}<ck^{4} \|f\|_{L^{2
   }(R_{k,2k})}
   \sup_{\vec k \in R_{k,2k }}\left( |Q_{n}|^{-1/2}\tilde N_{n}^{4}\sup _{\vec k \in R_{k,2k}}\|\Psi _n(\vec k,\cdot)\|_{W^{4}_2(Q_n)}\right).$$
Similarly,
    \begin{multline*}
     \|(S_{n+1}-S_n)f\|_{L^2(\R^2)} < \\
    c k^{4}\|f\|_{L^{2 }(R_{k,2k})}
\sup_{\vec k \in R_{k,2k }}\left(     |Q_{n+1}|^{-1/2}\tilde N_{n+1}^{4}
    \sup _{\vec k\in R_{k,2k}}
    \|\bigl(\Psi _{n+1}(\vec k, \cdot)- \Psi _n(\vec k,\cdot)\bigr)\|_{W^{4}_2
    (Q_{n+1})}\right).
    \end{multline*}
    It is proven in \cite{KL3} (Section 6.2) that
     \begin{equation}\|\Psi _{n+1}(\vec k, \cdot)- \Psi _n(\vec k, \cdot)\|_{L^2(Q_{n+1})} <
    c\epsilon_n^{3}|Q_{n+1}|^{1/2},\ \ n \geq 1, \ \epsilon _n=e^{-\frac 14k^{\eta s_n}}, \
    \mbox{when}\ \ \vec k \in R_{k,2k}.
    \label{s++}
    \end{equation} Applying the equation for eigenfunctions twice, we arrive to:
     \begin{equation}\|\Psi _{n+1}(\vec k, \cdot)- \Psi _n(\vec k, \cdot)\|_{W_2^{4}(Q_{n+1})} <
    ck^{4}\epsilon_n^{3}|Q_{n+1}|^{1/2},\ \ \ n \geq 1, \ \
    \mbox{when}\ \ \vec k \in R_{k,2k}.
    \label{s++app}
    \end{equation} Using the last estimate, we obtain
    \begin{equation} \|(S_n-S_{n+1})f\|_{L^2(\R^2)}
    \leq ck^{8}\|f\|_{L^{2}(R_{k,2k})}
   \sup _{\vec k \in R_{k,2k}}\left( \tilde N_{n}^{4} \epsilon_n^{3}\right).
    \end{equation}
    when the support of $f$ is in $R_{k,2k}$.
Considering that $\epsilon_n$ decays super-exponentially with $n$
(see the formula above  (\ref{eigenvalue-n*}) and the estimate $\tilde N_{n}\approx k^{s_n}$, we
conclude that
\begin{equation} \|(S_n-S_{n+1})f\|_{L^2(\R^2)}
    \leq c\|f\|_{L^{2}(R_{k,2k})}\epsilon_n^{2}(k),
  %%%%%% sup _{\vec k \in R_{k,2k}}\left( \tilde N_{n}^{4} \epsilon_n^{3}\right).
  \label{f1}
    \end{equation}
i.e., $S_nf $ is a Cauchy sequence in $L^2(\R^2)$ for every
$f \in L^{2}\left(R_{k,2k}\right)$.

If $f \in L^{2}({\mathcal G}_{\infty ,\lambda })$, then we can
express it as a sum of  functions $f _k $ such that $f _k $ has
support in $R_{k,2k}$. %%%and $\|f\|_{L^{\infty}(R_{k,2k})}\leq
%%\|f\|_{L^{\infty }({\mathcal G}_{\infty ,\lambda})}$.
Summing up estimates
(\ref{f1}) over all $k$ and using the Cauchy-Schwarz inequality on the right, we easily see that:
\begin{equation} \|(S_n-S_{n+1})f\|_{L^2(\R^2)}
    \leq c\|f\|_{L^{2}({\mathcal G}_{\infty ,\lambda })}\epsilon_n^{}(k_{*}),\ \ \ n\geq 1,
  %%%%%% sup _{\vec k \in R_{k,2k}}\left( \tilde N_{n}^{4} \epsilon_n^{3}\right).
  \label{f1a}
    \end{equation}
i.e., $S_n$ is a Cauchy
sequence in the space of bounded operators. We denote the limit of $S_n({\mathcal G}_{\infty
,\lambda })f $ by $S_{\infty }({\mathcal G}_{\infty ,\lambda })f $. By Theorem 2.3 in \cite{KL3}
\begin{equation} \|(S_0-S_{1})\|_{L^2(\R^2)}
<\lambda_*^{-\gamma _6}. \end{equation}
Estimate \eqref{June5} easily follows.
\end{proof}

\subsection{Appendix 2 (Proof of Lemma 2.2)} \label{App2}
Using \eqref{raz*}, \eqref{raz**}
 and integrating by parts $j$ times in  \eqref{eq2}, we obtain:
  \begin{eqnarray} \label{T_n}
  \left| (T_n F)(\vec k) \right| & \leq & \left(\sum_{r: |\vec k +\vec p_r^{(n)}| \geq |\vec k/4|} \frac{|C_r^{(n)}(\vec k)|}{|\vec k + \vec p_r^{(n)}|^j}\right) \|F\|_{W^1_j(\R^2)} \nonumber \\
  & & + \left(\sum_{r: |\vec k +\vec p_r^{(n)}| < |\vec k/4|} |C_r^{(n)}(\vec k)|\right) \|F\|_{L^1(\R^2)}.
  \end{eqnarray}
 Noting that
  \begin{equation}\label{C_r}
  \sum_r  |C_r^{(n)}(\vec k)| < 2,
  \end{equation}
 we can estimate the first term in the right hand side of \eqref{T_n}, i.e.,
  \begin{equation}\label{C_r bound large}
  |\vec k|^j \left( \sum_{r: |\vec k +\vec p_r^{(n)}| \geq |\vec k/4|} \frac{|C_r^{(n)}(\vec k)|}{|\vec k + \vec p_r^{(n)}|^j}\right) \|F\|_{W^1_j(\R^2)}
  \leq 2^{2j+2}  \|F\|_{W^1_j(\R^2)}.
  \end{equation}
 %from a similar calculation above using (6.10) in \cite{KL1}
 Next, since $H\Psi_n =\lambda_n\Psi_n$, we have
  \begin{equation}\label{innerproduct}
  \left(H\Psi_n, \f{e^{i \la \vec k+\vec p_r^{(n)},\cdot \ra}}{|Q_n|}\right) = \lambda_n(\vec k)C_r^{(n)}(\vec k).
  \end{equation}
 Note that
  $$
  \left(H\Psi_n, \f{e^{i \la \vec k+\vec p_r^{(n)},\cdot \ra}}{|Q_n|}\right)
  =\left(\Psi_n, \f{|\vec k+\vec p_r^{(n)}|^2 e^{i \la \vec k+\vec p_r^{(n)},\cdot \ra}}{|Q_n|}\right)
  +\left(\Psi_n, W_n\f{e^{i \la \vec k+\vec p_r^{(n)},\cdot \ra}}{|Q_n|}\right).
  $$
 Since the length of $V_r$ grows at most linearly with period, i.e., if $|\vec p_r^{(n)}-\vec p_{r'}^{(n)}| \geq R_0$, then $(W_n)_{r-r'}=0$,
 we get
  $$
  \left(\Psi_n, W_n\f{e^{i \la \vec k+\vec p_r^{(n)},\cdot \ra}}{|Q_n|}\right)
  = \sum_{r':|\vec p_{r'}^{(n)}-\vec p_r^{(n)}| < R_0} \left(\Psi_n, (W_n)_{r-r'}\f{e^{i \la \vec k+\vec p_{r'}^{(n)},\cdot \ra}}{|Q_n|}\right).
  $$
 Hence,
  \begin{equation}\label{recurse}
  C_r^{(n)}(\vec k)=(\lambda_n(\vec k)- |\vec k+\vec p_r^{(n)}|^2 )^{-1} \sum_{r':|\vec p_{r'}^{(n)}-\vec p_r^{(n)}| < R_0} (W_n)_{r-r'} C_{r'}^{(n)}(\vec k).
  \end{equation}
 and therefore
  \begin{align*}
  \sum_{r: |\vec k +\vec p_r^{(n)}| < |\vec k/2|} |C_r^{(n)}(\vec k)|
  & \leq \sum_{r: |\vec k +\vec p_r^{(n)}| < |\vec k/2|} \sum_{r':|\vec p_{r'}^{(n)}-\vec p_r^{(n)}| < R_0} \frac{|(W_n)_{r-r'}| |C_{r'}^{(n)}(\vec k)|}{|\lambda_n(\vec k)-|\vec k+\vec p_r^{(n)}|^2|}\\
  & \leq \frac{C}{|\vec k|^2} \sum_{r, r'} |(W_n)_{r-r'}| |C_{r'}^{(n)}(\vec k)| \leq \frac{C(V)}{|\vec k|^2}.
  \end{align*}
 By a recursive argument, while $j < k_{**}/(4R_0)$, we obtain
  \begin{equation}\label{C_r bound small}
  \sum_{r: |\vec k +\vec p_r^{(n)}| < |\vec k/4|} |C_r^{(n)}(\vec k)|
  \leq \frac{C(j, V)}{|\vec k|^j}.
  \end{equation}
 Using \eqref{C_r bound large} and \eqref{C_r bound small} in \eqref{T_n}, we obtain
  $$
  \left||\vec k|^j (T_n F)(\vec k)\right| \leq C(j,V,F).
  $$
 %Next, we consider the case with $|\alpha|=1$. Note that
  %$$
 % \nabla_{\vec k} \left(F(\vec x), \Psi_n(\vec k, \vec x)\right)
  %= \left(-i\vec x F(\vec x), \Psi_n(\vec k, \vec x)\right)
  %+ \left( F(\vec x), e^{i \la \vec k, \vec x \ra} \nabla_{\vec k} u_n(\vec k, \vec x)\right)
  %$$
 %and that
  %$$
  %\nabla_{\vec k} u_n(\vec k, \vec x)=\sum_{r \in \Z^2} \left(\nabla C_r^{(n)}e^{i \la \vec p_r^{(n)},\vec x %\ra}\right).
 % $$
% From $\|\nabla_{\vec k} (u_{n+1}-u_n)\|_{L^\infty_{\vec x}} \leq \eps_n^2(k)$
 %and $\|\nabla_{\vec k} (u_1-u_0)\|_{L^\infty_{\vec x}} \leq k^{-\gamma_0+\gamma_1}$, we get
  %$$
 % \|\nabla_{\vec k} u_n(\vec k, \vec x)\|_{L^\infty_{\vec x}} \leq Ck^{-\gamma_0+\gamma_1}
  %$$
 %and therefore
  %$$
 % \left| \left(F(\vec x), \nabla_{\vec k} \Psi_n(\vec k, \vec x)\right) \right| \leq C \|\vec x F(\vec x)\|_{L^1(\R^2)}.
  %$$
 %Now, we estimate $\left||\vec k|^j \left(F(\vec x), \nabla_{\vec k} \Psi_n(\vec k, \vec x)\right)\right|$ using an %analog of \eqref{T_n}.
 %We consider that $\left|\nabla C_r^{(n)}(\vec k)\right| < k^{-\gamma_0+\gamma_1}$.
 %Differentiating \eqref{innerproduct}, we obtain the estimate analogous to \eqref{C_r bound small}.
 %It will give us \eqref{T_infty bound} when $|\alpha|=1$.
 %Similarly, one can prove estimates for all $\alpha$.
The case $|m|>0$ can be covered using integration by parts. To obtain the decay rate for $D^m(C_r^{(n)}(\vec k)\eta_n(\vec k))$ one differentiates the recursive version of \eqref{recurse} and uses a priori estimates for $\sum_r|D^m(C_r^{(n)}(\vec k)\eta_n(\vec k))|$ (see \cite{KL3} and  (\ref{raz*}), \eqref{raz**}).

%%%%% \end{proof}

\subsection{Appendix 3 (Proof of \eqref{ineq:I_3})}

 We consider the limit-periodic case. A proof for the quasi-periodic case is analogous. Indeed, we can write $u_\infty(\vec k,\vec x)$ as follows:%(here we need proper notation)
 $$
 u_\infty(\vec k,\vec x)=\sum_{n=1}^\infty \wti u_n(\vec k, \vec x),
 $$
where $ \wti u_n(\vec k, \vec x)= \sum_{r \in \Z^2} \wti C_r^{(n)} e^{i\la {\vec p_r}^{(n)},\,\vec x\ra}$,
$\vec p_r^{(n)}$ are vectors of the dual lattice corresponding to $W_n$.
We obtain:
 \begin{align*}
 &\| \phi_s \|^2_{L^2(B_R)} \le \| \phi_s \|_{L^2(\R^2)}^2 \\
 &= \frac{1}{(2\pi)^2}\int_{\R^2} \int_{\wti {\mathcal G}_\infty}\!\!\int_{\wti {\mathcal G}_\infty}
  e^{i\la \vec k- \vec \xi, \vec x\ra} \la \nabla_{\vec k} \wti u_\infty(\vec k,\vec x), \nabla_{\vec \xi} \wti u_\infty(\vec \xi,\vec x)\ra e^{-it\bigl(\lambda_\infty(\vec k)-\lambda_\infty(\vec \xi)\bigr)}\\
 &\hskip 3cm \hatt{\Psi}_0(\vec k)\ol{\hatt{\Psi}_0(\vec \xi)}\eta_\delta(\vec k)\eta_\delta(\vec \xi)\,d\vec k\,d\vec \xi \,d\vec x\\
 &=\frac{1}{ (2\pi)^2}\sum_{n,m \in \N}\sum_{r,q \in \Z^2}% \in \Z^2 \setminus \{0\}}
  \int_{\R^2} \int_{\wti {\mathcal G}_\infty}\!\!\int_{\wti {\mathcal G}_\infty}
  e^{i\la \vec k- \vec \xi +\vec p_r^{(n)}-\vec p_q^{(m)}, \vec x\ra} \la \nabla  \wti C_r^{(n)}(\vec k), \nabla \wti C_q^{(m)}(\vec \xi)\ra\\ & \hskip 3cm e^{-it\bigl(\lambda_\infty(\vec k)-\lambda_\infty(\vec \xi)\bigr)} \hatt{\Psi}_0(\vec k)\ol{\hatt{\Psi}_0(\vec \xi)}\eta_\delta(\vec k)\eta_\delta(\vec \xi)\,d\vec k\,d\vec \xi \,d\vec x\\
 &= \sum_{n,m \in \N}\sum_{r,q \in \Z^2}% \in \Z^2 \setminus \{0\}}
  \int_{\wti {\mathcal G}_\infty}\!\!\int_{\wti {\mathcal G}_\infty}
   \delta(\vec k- \vec \xi +\vec p_r^{(n)}-\vec p_q^{(m)})\la \nabla  \wti C_r^{(n)}(\vec k), \nabla \wti C_q^{(m)}(\vec \xi)\ra \\
   &\hskip 3cm e^{-it\bigl(\lambda_\infty(\vec k)-\lambda_\infty(\vec \xi)\bigr)} \hatt{\Psi}_0(\vec k)\ol{\hatt{\Psi}_0(\vec \xi)}\eta_\delta(\vec k)\eta_\delta(\vec \xi)\,d\vec k\,d\vec \xi\\
 &=\sum_{n,m \in \N}\sum_{r,q \in \Z^2}% \in \Z^2 \setminus \{0\}}
  \int_{\wti {\mathcal G}_\infty \cap  (\wti{\mathcal G}_\infty -\vec p_r^{(n)}+\vec p_q^{(m)})}\la \nabla  \wti C_r^{(n)}(\vec k), \nabla \wti C_q^{(m)}(\vec k +\vec p_r^{(n)}-\vec p_q^{(m)})
  \ra\\
 &\hskip 0.5cm  e^{-it\bigl(\lambda_\infty(\vec k)-\lambda_\infty(\vec k +\vec p_r^{(n)}-\vec p_q^{(m)})\bigr)}\hatt{\Psi}_0(\vec k)\ol{\hatt{\Psi}_0(\vec k +\vec p_r^{(n)}-\vec p_q^{(m)})}\eta_\delta(\vec k)\eta_\delta(\vec k +\vec p_r^{(n)}-\vec p_q^{(m)})\,d\vec k\\
 &\le \frac{1}{2}\sum_{n,m \in \N}\sum_{r,q \in \Z^2}% \in \Z^2 \setminus \{0\}}
  \int_{\wti {\mathcal G}_\infty \cap (\wti {\mathcal G}_\infty -\vec p_r^{(n)}+\vec p_q^{(m)})}|\nabla  \wti C_r^{(n)}(\vec k)|\ |\nabla \wti C_q^{(m)}(\vec k +\vec p_r^{(n)}-\vec p_q^{(m)})| \ |\hatt{\Psi}_0(\vec k)|^2 \,d\vec k \\
 &+ \frac{1}{2} \sum_{n,m \in \N}\sum_{r,q \in \Z^2}% \in \Z^2 \setminus \{0\}}
  \int_{\wti {\mathcal G}_\infty \cap (\wti {\mathcal G}_\infty -\vec p_r^{(n)}+\vec p_q^{(m)})}|\nabla  \wti C_r^{(n)}(\vec k)|\ |\nabla \wti C_q^{(m)}(\vec k +\vec p_r^{(n)}-\vec p_q^{(m)})| \\
  & \hskip 3cm |\hatt{\Psi}_0(\vec k +\vec p_r^{(n)}-\vec p_q^{(m)})|^2 \,d\vec k \\
 &=  \int_{\wti {\mathcal G}_\infty} \left[\sum_{n,m \in \N}\sum_{r,q \in \Z^2} |\nabla  \wti C_r^{(n)}(\vec k)|\ |\nabla \wti C_q^{(m)}(\vec k +\vec p_r^{(n)}-\vec p_q^{(m)})|\chi_{\wti {\mathcal G}_\infty} (\vec k +\vec p_r^{(n)}-\vec p_q^{(m)}) \right]|\hatt{\Psi}_0(\vec k)|^2\,d\vec k \\
 &\le c_4(V)\|\Psi_0\|_{L^2(\R^2)}^2,
 \end{align*}
since
$ \sum_{n,m \in \N}\sum_{r,q \in \Z^2} |\nabla  \wti C_r^{(n)}(\vec k)|\ |\nabla \wti C_q^{(m)}(\vec k +\vec p_r^{(n)}-\vec p_q^{(m)})|\chi_{\wti {\mathcal G}_\infty} (\vec k +\vec p_r^{(n)}-\vec p_q^{(m)})$ is bounded uniformly in $\vec k\in \wti {\mathcal G}_\infty$, say, it is bounded by $c_4(V)$, see \cite{KL3}.


\begin{thebibliography}{100}

\bibitem{AK} J.~Asch and A.~Knauf,
{\em Motion in periodic potentials},
Nonlinearity 11 (1998), 175--200

\bibitem{AS} J.~Avron and B.~Simon, {\em Almost periodic Schr\"odinger operators, I. Limit periodic potentials}, Comm.\ Math.\ Phys.\ 82 (1981), 101--120


\bibitem{BSB} J.~Bellissard and H.~Schulz-Baldes, {\em Subdiffusive quantum transport for 3D Hamiltonians with absolutely continuous spectra}, J.\ Stat.\ Phys.\ 99 (2000), 587--594

\bibitem{B2007} J.~ Bourgain, {\em Anderson localization for quasi-periodic lattice Schr\"odinger operators on $\Z^d$, $d$ arbitrary}, Geom.\ Funct.\ Anal.\ 17 (2007), 682--706

\bibitem{BGS} J.~Bourgain, M.~Goldstein and W.~Schlag, {\em Anderson localization for Schr\"odinger operators on $\Z^2$ with quasi-periodic potential}, Acta Math.\ 188 (2002), 41--86

\bibitem{ChD} V.~A.~Chulaevsky and E.~I.~Dinaburg, {\em Methods of KAM-theory
for long-range quasi-periodic operators on $\Z^{\nu}$. Pure point spectrum}, Commun.\ Math.\ Phys.\ 153 (1993), 559--577

\bibitem{CD} V.~Chulaevsky and F.~Delyon, {\em Purely absolutely continuous spectrum for almost Mathieu operators}, J.\ Stat.\ Phys.\ 55 (1989), 1279--1284

\bibitem{C} J.-M.~Combes,
{\em Connections between quantum dynamics and spectral properties of time-evolution operators},
Differential equations with applications to mathematical physics,
Mathematics in Science and Engineering, 192 (1993), 59--68

\bibitem{DLS} D.~Damanik, D.~Lenz and G.~Stolz, {\em Lower transport bounds for one-dimensional continuum Schr\"odinger operators}, Math.\ Ann.\ 336 (2006), 361--389


\bibitem{DT} D.~Damanik and S.~Tcheremchantsev, {\em Power-law bounds on transfer matrices and quantum dynamics in one dimension}, Comm.\ Math.\ Phys.\ 236 (2003), 513--534

\bibitem{DT2} D.~Damanik and S.~Tcheremchantsev, {\em Scaling estimates for solutions and dynamical lower bounds on wavepacket spreading}, J.\ Anal.\ Math. 97 (2005), 103--131

\bibitem{DT3} D.~Damanik and S.~Tcheremchantsev, {\em Upper bounds in quantum dynamics}, J.\ Amer.\ Math.\ Soc.\ 20 (2007), 799-827

\bibitem{DT2010} D.~Damanik and S.~Tcheremchantsev, {\em A general description of quantum dynamical spreading over an orthonormal basis and applications to Schr\"odinger operators}, Discrete Contin.\ Dyn.\ Syst.\ 28 (2010), 1381--1412 

\bibitem{DS} E.~I.~Dinaburg and Ya.~Sinai, {\em The one-dimensional Schr\"odinger equation with a quasi-periodic potential}, Funct.\ Anal.\ Appl.\ 9 (1975), 279--289

\bibitem{El} L.~H.~Eliasson, {\em Floquet solutions for the 1-dimensional quasi-periodic Schr\"odinger equation}, Comm.\ Math.\ Phys.\ 146 (1992), 447--482

\bibitem{Ge} I.~M.~Gel'fand,
{\em Expansion in Eigenfunctions of an Equation with Periodic Coefficients.}
Dokl. Akad. Nauk SSSR, 73 (1950), 1117--1120 (in Russian).

\bibitem{GKT} F.~Germinet, A.~Kiselev and S.~Tcheremchantsev,
{\em Transfer matrices and transport for Schr\"{o}dinger operators},
Ann.\ Inst.\ Fourier 54 (2004), 787--830


\bibitem{G} I.~Guarneri,
{\em Spectral properties of quantum diffusion on discrete lattices},
Europhysics Lett.\ 10 (1989), 95--100

\bibitem{G2} I.~Guarneri, {\em On an estimate concerning quantum diffusion in the presence of a fractional spectrum}, Europhys.\ Lett.\ 21 (1993), 729--733


\bibitem{H1} L.~H\"ormander,
{\it The analysis of linear partial differential operators. I. Distribution theory and Fourier analysis},
Springer, 256 (1990)

\bibitem{JSBS} S.~Jitomirskaya, H.~Schulz-Baldes and G.~Stolz, {\em Delocalization in random polymer models}, Comm.\ Math.\ Phys.\ 233 (2003), 27--48

\bibitem{JM} R.~Johnson and J.~Moser, {\em The rotation number for almost periodic potentials}, Comm.\ Math.\ Phys.\ 84 (1982), 403--438

\bibitem{KL1} Yu.~Karpeshina and Y.-R.~Lee,
{\em Spectral properties of polyharmonic operators with limit-periodic potential in dimension two},
J.\ Anal.\ Math.\ 102 (2007), 225--310

\bibitem{KL2} Yu.~Karpeshina and Y.-R.~Lee,
{\em Absolutely continuous spectrum of a polyharmonic operator with a limit-periodic potential in dimension two},
Comm.\ Partial Differential Equations 33 (2008), 1711--1728

\bibitem{KL3} Yu.~Karpeshina and Y.-R.~Lee,
{\em Spectral properties of a limit-periodic Schr\"odinger operator in dimension two},
J.\ Anal.\ Math.\ 120 (2013), 1--84

\bibitem{KS} Yu.~Karpeshina and R.~Shterenberg,
{\em Multiscale analysis in momentum space for quasi-periodic potential in dimension two},
J.\ Math.\ Phys.\ 54, 073507 (2013), 1--92

\bibitem{KS1} Yu.~Karpeshina and R.~Shterenberg, {\em Extended States for the Schr\"odinger Operator with Quasi-periodic Potential in Dimension Two}, arXiv:1408.5660 (2014)

\bibitem{KL} A.~Kiselev and Y.~Last,
{\em Solutions, spectrum, and dynamics for Schr\"odinger operators on infinite domains},
Duke Math.\ J.\ 102 (2000), 125--150

\bibitem{L} Y.~Last,
{\em Quantum dynamics and decompositions of singular continuous spectra},
J.\ Funct.\ Anal.\ 142 (1996), 406--445

\bibitem{MC} S.~A.~Molchanov and V.~Chulaevsky, {\em The structure of a spectrum of lacunary-limit-periodic Schr\"odinger operator}, Functional Anal.\ Appl.\ 18 (1984), 343--344

\bibitem{MP} J.~Moser and J.~P\"oschel, {\em An extension of a result by Dinaburg and Sinai on quasiperiodic potentials}, Comment.\ Math.\ Helv.\ 59 (1984), 39--85

\bibitem{PF} L.~Pastur and A.~Figotin, {\em Spectra of random and almost-periodic operators}, Springer, 1992

\bibitem{RaSi} C.~Radin and B.~Simon,
{\em Invariant domains for the time-dependent Schr\"odinger equation},
J.\ Differential Equations 29 (1978), 289--296

\bibitem{ReedSimon}  M.~Reed and B.~Simon, {\em Methods of modern mathematical physics. II. Fourier analysis, self-adjointness}. Academic Press, New York-London, 1975

\bibitem{R} H.~R\"ussmann, On the one dimensional Schr\"odinger equation with quasi-periodic potential, Ann.\ N.\ Y.\ Acad.\ Sci.\ 357 (1980), 90--107

\bibitem{SkSo}  M.~M.~ Skriganov and A.~V.~Sobolev,
{\em On the spectrum of a limit-periodic Schr\"odinger operator},
Algebra i Analiz, 17 (2005), 5;
English translation: St.\ Petersburg Math.\ J.\ 17 (2006), 815--833

\bibitem{Tch} S.~Tcheremchantsev,
{\em Mixed lower bounds for quantum transport},
J.\ Funct.\ Anal.\ 197 (2003) 247--282

\end{thebibliography}
\end{document}